%% file: main.tex
\documentclass[11pt, twoside, a4paper]{article}

\input{commands}

\date{\today}

\newcommand{\blue}[1]{{\color{blue}#1}}

\title{The Spear and the Ring: Emergent Structures in Magnetic Colloidal Suspensions}

\author[1,5]{Raphaël Côte}
\author[1]{Clémentine Courtès}
\author[2]{Guillaume Ferriere}
\author[3]{Ludovic Godard-Cadillac}
\author[4,6]{Yannick Privat}

\affil[1]{\footnotesize IRMA, Universit\'e de Strasbourg, CNRS UMR 7501, Inria, 7 rue Ren\'e Descartes, 67084 Strasbourg (France)\smallskip}
\affil[2]{\footnotesize Univ. Lille, Inria, CNRS, UMR 8524 - Laboratoire Paul Painlevé, F-59000 Lille (France)\smallskip}
\affil[3]{\footnotesize Univ. Bordeaux, Bordeaux INP, Institut de Mathématiques de Bordeaux, UMR 5251, F-33400 Talence (France)\smallskip}
\affil[4]{\footnotesize Université de Lorraine, CNRS, Inria SPHINX, IECL, F-54000 Nancy (France)\smallskip}
\affil[5]{\footnotesize Institut d’études avancées de l’université de Strasbourg (USIAS)\smallskip}
\affil[6]{\footnotesize Institut Universitaire de France (IUF)}

\begin{document}

\maketitle

\blue{
\begin{abstract}
We study the magnetic-colloid model introduced in~\cite{Klughertz_2016_these} and numerically investigated in~\cite{Godard-Cadillac_Hervieux_Manfredi_2024}. We analyze two stationary configurations observed in these simulations: the spear and the ring.

For the spear, a chain of aligned particles, the energy reduces to a one-dimensional long-range Lennard--Jones-type interaction. We prove existence of a minimizing configuration within the spear ansatz and uniqueness up to translations under an explicit condition on the repulsive exponent. We obtain uniform bounds on the neighboring-particle distances and quantitative finite-$N$ estimates describing their convergence toward the equilibrium spacing in the bulk. These estimates also exhibit the influence of the endpoints and show that the convergence is not uniform up to the boundary. For the ring, where the particles are uniformly distributed on a circle with tangential magnetic moments, we reduce the critical-point condition to a scalar equation. This yields a unique critical radius, an explicit formula for it, and its large-$N$ asymptotics.
\end{abstract}
}
\setcounter{tocdepth}{2}
\tableofcontents

\section*{Introduction}
Magnetic spherical nanoparticles have attracted significant interest over the past decades, both for their theoretical implications and diverse applications. Their potential uses span a wide range of fields, including imaging (microscopic imaging), high-density data storage (leveraging anisotropy properties), targeted medical therapy (precise drug delivery), ferrofluids and microfluidics, nano\-swimmers, and more.
For more details on potential applications, we refer to the Ph.D. thesis~\cite{Klughertz_2016_these} or to related work on magnetic nanoparticle control~\cite{Cote_Courtes_Ferriere_Privat_2024, Agarwal_Carbou_Labbe_Prieur_2011}.

\blue{In this paper, we mathematically analyze the nanoparticle model introduced in this thesis, aiming to characterize stationary structures predicted by the model and observed in numerical simulations.  

The model was developed by G.~Klughertz, P.-A.~Hervieux and G.~Manfredi in \cite{Klughertz_2016_these}.
} 
The recent numerical study \cite{Godard-Cadillac_Hervieux_Manfredi_2024} has documented the filamentary behaviour and the spear and ring configurations that form the basis of our investigation. 
Our contribution consists in analysing these structures within the existing modelling framework through a rigorous variational approach.

This model is based on the following approximations:  
\begin{itemize}
    \item[(i)] Nanoparticles are assumed to be spherical.  
    \item[(ii)] The magnetization within each particle is considered constant and replaced by a single \textit{super-spin} at its center (a valid approximation for nanoscale metallic ferromagnetic particles).  
    \item[(iii)] The nanoparticle's spin is assumed to be fixed in its own frame, as anisotropy effects dominate over the Landau-Lifshitz-Gilbert equation.  
\end{itemize}

\blue{
Under these assumptions, a nanoparticle is described by the position of its center and the orientation of its magnetic moment. The particles interact through the standard dipole--dipole magnetic interaction, which acts on both positions and orientations and favors head-to-tail alignment.

A short-range radial repulsion is added to prevent particle overlap. The corresponding effective exclusion distance may include the protective coating of the particles.

The viscous terms dissipate mechanical energy. The simulations in~\cite{Klughertz_2016_these,Godard-Cadillac_Hervieux_Manfredi_2024} show filamentary configurations and, for relatively small numbers of particles, two configurations that will be studied below: the \textit{spear} and the \textit{ring}. We use these numerical observations as motivation and do not prove convergence of the full dissipative dynamics.
}  
The spear structure consists of particles aligned along a single direction, with all spins oriented identically.  
In contrast, the ring structure features particles evenly distributed along a circle, with spins tangent to the circle and aligned in the same rotational sense.  
A transition from the spear to the ring can occur by continuously bending the nanoparticle chain until its two ends meet~\cite{Godard-Cadillac_Hervieux_Manfredi_2024}.  

The aim of this work is to provide a mathematical characterization of these two structures, the spear and the ring, within the framework of the calculus of variations, and to establish asymptotic results for a large number of particles.  
The article is structured into two main parts: the first focuses on the study of the spear, while the second examines the ring.

\section{Presentation of the problem and main results}

\subsection{A model for magnetic nanoparticles colloidal suspensions}

The model proposed in Klughertz's thesis~\cite{Klughertz_2016_these} for colloidal suspensions is a simple yet rich model, which consists in magnetic spherical solids immersed in a viscous fluid. 
One magnetic nanoparticle is represented as an element 
\begin{equation} (x,m) \in \RR^3\times\SS^2. \end{equation}
The first three coordinates of $ x \in \mathbb{R}^3 $ represent the position of the nanoparticle's center, while the magnetic moment  
$ m \in \mathbb{S}^2 $  
\blue{is normalized to norm $1$.
This model is derived under the approximation of constant magnetization within the spherical nanoparticle,  
a standard assumption in the context of nano-scale structures; see, for instance,~\cite{Allouges_Beauchard_2009}.  

We now consider a system of $N\geq 2$ nanoparticles $ (x_k, m_k) \in \mathbb{R}^3 \times \mathbb{S}^2 $ for $ k = 1, \dots, N $ and describe their interactions.  
For notational simplicity, we define  
}  
\begin{equation}
(X, M) := \big((x_1, m_1), \dots, (x_N, m_N)\big) \in (\mathbb{R}^3 \times \mathbb{S}^2)^{N}.  
\end{equation}

\subsubsection{Influence of the magnetic interaction}
The interaction between two spherical nanoparticles in a fluid is governed by the standard dipole-dipole magnetic interaction.  
We recall the formula for the magnetic field $\mathcal H^d$ generated by a magnetic dipole $ m \in \mathbb{S}^2 $ centered at $ (0,0,0) $ (after renormalization of physical constants):  
\begin{equation}
    \mathcal H^d (x) := \frac{3(m\cdot e_r)e_r - m}{|x|^3},
\end{equation}
where $ e_r := x/|x| $ is the unit radial vector.  

For a given nanoparticle $ (x_k, m_k) \in \mathbb{R}^3 \times \mathbb{S}^2 $, we denote by $ \mathcal  H_k $ the magnetic field it generates.  
The total magnetic field $ \mathcal H $ produced by a system of nanoparticles $ (x_k,m_k) $ for $ k=1,\dots,N $ is the sum of all individual dipolar fields $ \mathcal H_k $.  

The interaction potential, representing the potential energy associated with the interaction between two nanoparticles $ (x_k, m_k) $ and $ (x_\ell, m_\ell) $, is given by:  
\begin{equation}
    U^d_{k \ell}(X,M) := -m_\ell \cdot \mathcal H_k = -m_k \cdot \mathcal H_\ell = \frac{m_k \cdot m_\ell}{|r_{k \ell}|^3} - 3\frac{(m_k \cdot r_{k \ell})(m_\ell \cdot r_{k \ell})}{|r_{k \ell}|^5},
\end{equation}
where  
\begin{equation} 
    r_{k\ell} = x_k - x_\ell.  
\end{equation}
One can verify that this quantity is minimized, for fixed values of $ x_k $ and $ x_\ell $, when $ m_k = m_\ell = \pm\frac{r_{k\ell}}{|r_{k\ell}|}$.  

To generalize this interaction and establish connections with the theory of Lennard-Jones potentials~\cite{Blanc_Lewin_2015,Canizo_Ramos-Lora:Discrete}, we consider the following slightly more general form of the magnetic potential:  
\begin{equation} \label{eq:interaction potential}
    U^d_{k\ell}(X,M) := B_0 \frac{m_k \cdot m_\ell}{|r_{k \ell}|^\beta} - (B+B_0) \frac{(m_k \cdot r_{k \ell})(m_\ell \cdot r_{k \ell})}{|r_{k \ell}|^{\beta+2}},
\end{equation}
where $B_0, B, \beta > 0$ are fixed parameters. We recover the previous expression by setting $B_0 = 1$, $B = 2$, and $\beta = 3$.  

The total potential energy of the nanoparticle system, associated with the magnetic dipolar interactions, is obtained by summing all pairwise interaction potentials $U^d_{k\ell}$ over all pairs of nanoparticles.  


\subsubsection{The repulsive potential}
\blue{
The magnetic interaction is singular when two particle centers approach each other. We therefore introduce a short-range repulsive interaction to exclude particle overlap. A standard idealization is the \textit{hard-sphere model}. After normalization of the prescribed center-to-center exclusion distance, it corresponds to the constraint
\begin{equation}
    \forall\;k\neq \ell,\qquad |x_k-x_\ell|\geq 1.
\end{equation}
}
Since this constraint can be hard to work with (from a mathematical point of view) but also questionable in terms of modeling (the nanoparticles are covered by a protective sheath, and the mechanical behavior of the hard sphere can be different), we rather work with the \emph{soft sphere model}. This model consists in adding a new potential $U^s_{k\ell}$ of the form
\begin{equation}
    U^s_{k \ell}(X,M):=A\frac{1}{|r_{k \ell}|^\alpha},
\end{equation}
\blue{
where $\alpha>0$ and $A>0$ are empirical parameters depending on the protective coating. We write $A=aR^\alpha$, where $R>0$ is an effective center-to-center exclusion distance and $a>0$ is the repellency coefficient. Thus, $R$ is not the physical radius of a nanoparticle.

We assume $\alpha>\beta$, so that the repulsive term dominates the attractive interaction at short distances. For fixed $r>0$, $a(R/r)^\alpha$ tends to $+\infty$ if $r<R$ and to $0$ if $r>R$ as $\alpha\to+\infty$. In this sense, large values of $\alpha$ approximate a hard-core exclusion at distance $R$.
}

\subsubsection{Potential energy and dynamics}
\blue{The dissipative dynamics considered in~\cite{Klughertz_2016_these,Godard-Cadillac_Hervieux_Manfredi_2024} decreases the mechanical energy, and the numerical simulations display relaxation toward stationary configurations. Here we only study critical points of the total potential energy and do not prove convergence of the dynamics:}
\begin{equation}\label{def:U}
    U(X,M) := U^d(X,M) + U^s (X,M),
\end{equation}
where $ U^d $ represents the total potential energy due to magnetic dipolar interactions, and $ U^s $ accounts for the total repulsive potential energy:  
\begin{equation}\label{def:U d s}
    U^d:=\sum_{k=2}^N\sum_{\ell=1}^{k-1}U_{k \ell}^d, \qquad
    U^s:=\sum_{k=2}^N\sum_{\ell=1}^{k-1}U_{k \ell}^s.
\end{equation}  

We also define the potential energy associated with the $ k $-th nanoparticle $ (x_k,m_k) $:  
\begin{equation}
    U_k:= U^d_k + U^s_k, \quad \text{with} \quad 
    U^d_k:=\sum_{\substack{\ell =1\\ \ell\neq k}}^{N}U_{k \ell}^d, \quad
    U^s_k:=\sum_{\substack{\ell=1\\ \ell \neq k}}^{N}U_{k \ell}^s.
\end{equation}  

Our goal is to investigate two specific critical points of this potential energy, which are frequently observed in simulations~\cite{Klughertz_2016_these,Godard-Cadillac_Hervieux_Manfredi_2024}:  
the {\it spear} (an aligned structure) and the {\it ring} (a circular structure).  

Although this work focuses on the study of the two steady-state structures, the {\it spear} and the {\it ring}, we present here the three equations governing the dissipative dynamics for completeness.  
For a more detailed analysis of the dynamics, we refer to~\cite{Klughertz_2016_these,Godard-Cadillac_Hervieux_Manfredi_2024}.  

\blue{
Dissipation in the system arises from viscous interactions with the surrounding fluid. The friction terms associated with translational and rotational motion (Stokes drag coefficients) are given by:  
\begin{equation}\label{def:friction terms}
    \zeta_{tr}=6\pi\nu \rho,\qquad\text{and}\qquad\zeta_r=8\pi\nu \rho^3,
\end{equation}
where $\rho$ is the physical radius of the nanoparticle and $\nu$ denotes the dynamic viscosity.
Newton’s second law gives:
}
\begin{equation}\label{eq:dynamics in translation}
\mu\frac{d^2}{dt^2} x_k = F_k -\zeta_{tr}\frac{d}{dt}x_k,\quad k=1,\dots,N
\end{equation}
where $\mu$ is the mass of the particle and where $F_k$ is the conservative contribution of the force applied on the $k^{th}$ nano-particle (computed using the gradient of the energy).
Similarly, the dynamics for the angular velocity vector $\omega_i$ is given by
\blue{
\begin{equation}\label{eq:dynamics in rotation}
    I\frac{d}{dt}\omega_k=T_k-\zeta_r\omega_k, \quad k=1,\dots,N
\end{equation}
with $T_k$ the conservative torque (computed with the energy), and where $I=2\mu \rho^2/5$ (the moment of inertia for the ball). The dynamics of the super-spin is given by the time-scale separation hypothesis (see~\cite{Klughertz_2016_these} for a discussion on this modeling hypothesis) which leads to
\begin{equation}\label{eq:time scale separation}
\frac{d}{dt}m_k=\omega_k\times m_k.
\end{equation}
Figures~\ref{fig:spear} and~\ref{fig:ring}, reproduced from~\cite{Godard-Cadillac_Hervieux_Manfredi_2024}, show two numerical evolutions leading respectively to configurations close to a spear and to a ring. The initial configuration is shown on the left, an intermediate state in the center and the last computed state on the right.
}

\begin{figure}[!ht]
    \centering
    \begin{subfigure}[b]{5.cm}
        \includegraphics[width=5cm, trim={3.5cm 1cm 1.5cm 1.5cm}, clip]{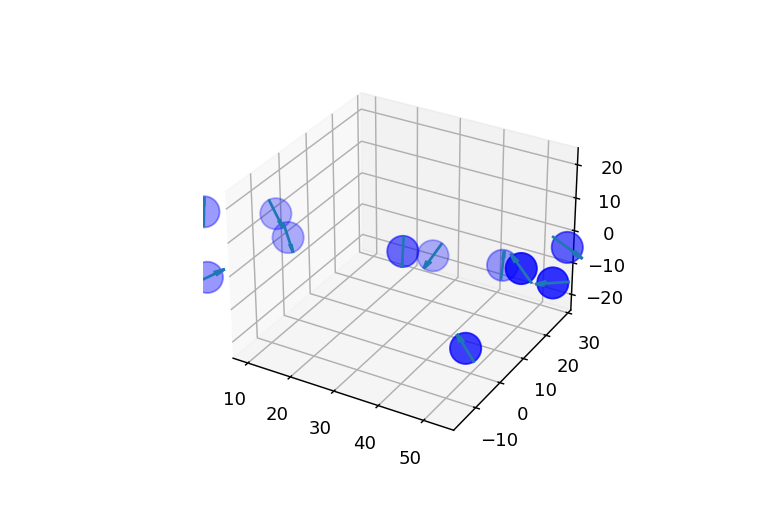}
    \end{subfigure}
    \begin{subfigure}[b]{5.cm}
        \includegraphics[width=5.cm, trim={3.5cm 1cm 1.5cm 1.5cm}, clip]{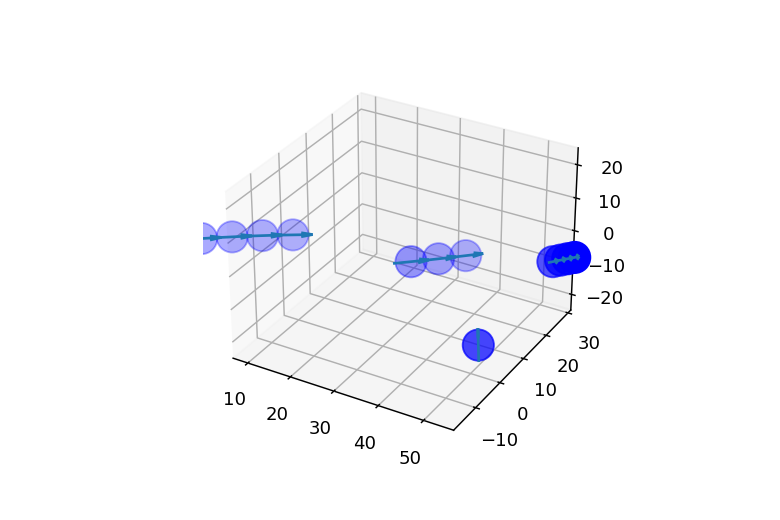}
    \end{subfigure}
    \begin{subfigure}[b]{5.cm}
        \includegraphics[width=5.cm, trim={3.5cm 1cm 1.5cm 1.5cm}, clip]{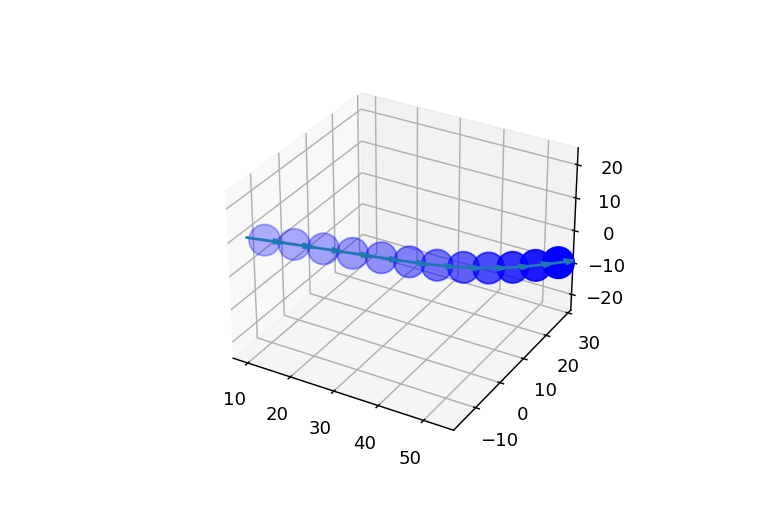}
    \end{subfigure}
    \caption{A dynamic for a system of $12$ magnetic nanoparticles (emergence of spear structure): snapshots at times $t=0,\; 1,\; 2$ $\mu$s.}\label{fig:spear}
\end{figure}
\begin{figure}[!ht]
    \centering
    \begin{subfigure}[b]{5.cm}
        \includegraphics[width=5.cm, trim={3.5cm 1cm 1.5cm 1.5cm}, clip]{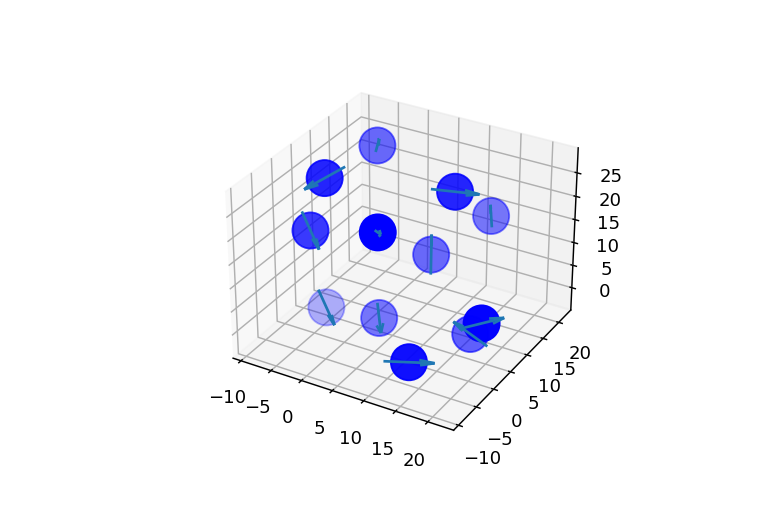}
    \end{subfigure}
    \begin{subfigure}[b]{5.cm}
        \includegraphics[width=5.cm, trim={3.5cm 1cm 1.5cm 1.5cm}, clip]{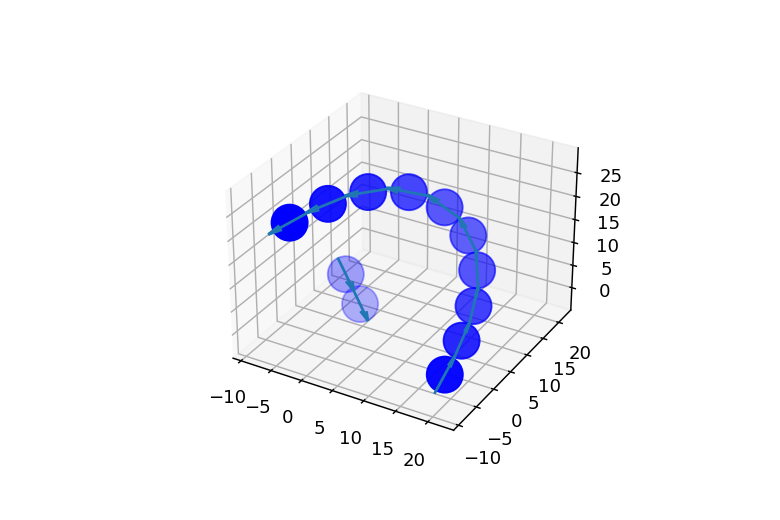}
    \end{subfigure}
    \begin{subfigure}[b]{5.cm}
        \includegraphics[width=5.cm, trim={3.5cm 1cm 1.5cm 1.5cm}, clip]{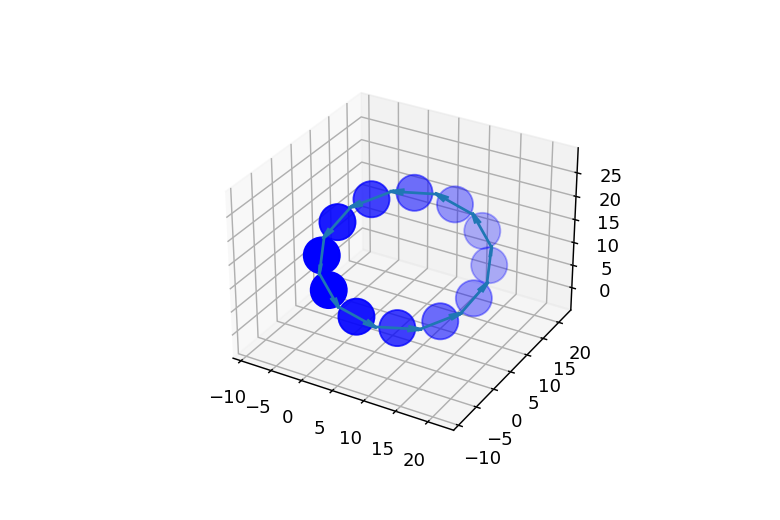}
    \end{subfigure}
    \caption{A dynamic for a system of $12$ magnetic nanoparticles (emergence of ring structure): snapshots at times $t=0,\; 1,\; 2$ $\mu$s.}\label{fig:ring}
\end{figure}
\subsection{The spear: An aligned structure}
The spear is an aligned structure emerging in this colloidal system (see Figure~\ref{fig:spear}). It is characterized by:
\begin{equation}\label{def:spear}
    x_k=\big(p_k,0,\,0\big),\qquad\text{and}\qquad m_k=\big(1,0,\,0\big),
\end{equation}
for $1 \leq k \leq N$ with $p_1\leq p_2\leq \dots\leq p_{N}$. 
Introducing $ P = (p_1, \dots, p_N) $, we denote this configuration as $ \mS_{P} \in (\mathbb{R}^3 \times \mathbb{S}^2)^N $.  
Numerical simulations suggest the existence of $ P \in \mathbb{R}^{N} $ such that the structure $ \mS_P $ is a local minimizer of the potential energy $ U $.  

Since the potential energy is translation-invariant, we assume without loss of generality that $ p_1 = 0 $.
In this article, we prove the existence of critical points that have a spear structure $\mS_P$. We also have a uniqueness result (up to translations). Finally, we study the distances between neighboring nanoparticles for fixed values of $N$ and in the asymptotic $N\to+\infty$. Three distances of interest arise from our work  (under the assumption $\alpha>\beta>1$) which are:
\begin{equation}\label{def:the H gang}
\widecheck{h}:=\bigg(\frac{\alpha A}{\beta B\zeta(\beta)}\bigg)^\frac{1}{\alpha-\beta},\qquad\overline{h}:=\bigg(\frac{\alpha A\zeta(\alpha)}{\beta B\zeta(\beta)}\bigg)^\frac{1}{\alpha-\beta}\qquad\text{and}\qquad\widehat{h}:=\bigg(\frac{\alpha A}{\beta B}\bigg)^\frac{1}{\alpha-\beta},
\end{equation}
where $\zeta(s) = \sum_{k \ge 1} \frac{1}{k^s}$ denotes the standard zeta function (for $s >1$). Remark that
\begin{equation} \widecheck{h}\;<\;\overline{h}\;<\;\widehat{h}. \end{equation}
These constants appear naturally in the study of $U$, which in the case of the spear, writes in terms of a Lennard-Jones potential \eqref{def:J L}.

\bigskip

In this article, we prove the following theorem.

\begin{theorem}[Existence and uniqueness of a critical spear structure]\label{thrm:spear}
Consider $U$ the magnetic-repulsive energy defined at~\eqref{def:U} with $\alpha>\beta>1$.
\begin{itemize}
    \item[(i)] There exists $P_0\in\RR^{N}$ such that the associated spear $\mS_{P_0}\in(\RR^3\times\SS^2)^N$ is a critical point of the energy $U$ in $( \RR^3 \times \SS^2 )^N$, and a minimizer of the problem
    \begin{equation}
        \inf_{P \in \RR^N} U (\mS_P).
    \end{equation}
    \item[(ii)] The distances between neighboring particles for any critical point of this form satisfy for all $k=1,\dots,N-1 $:
    \begin{equation} 
        \widecheck{h}\leq |x_{k}-x_{k+1}|\leq\widehat{h}.
    \end{equation}
    \item[(iii)] \blue{
    Let $\alpha_\ast(\beta)$ be the threshold defined explicitly in~\eqref{eq:def_alpha_master}. If $\alpha > \alpha_\ast(\beta)$, then the minimizing gap vector is unique; equivalently, the ordered minimizing configuration $P_0$ is unique up to translations. Furthermore, there exist positive constants $c=c(A,B,\alpha,\beta)$ and $C=C(A,B,\alpha,\beta)$, independent of $N$, such that for all $k=1,\dots,N-1$:
    \begin{equation}\label{eq:asymptotics}
        \overline{h}+\frac{c}{N^{\beta-1}}\leq|x_k- x_{k+1}|\leq \overline{h}+C\bigg(\frac{1}{k^{\beta-1}}+\frac{1}{(N-k)^{\beta-1}}\bigg).
    \end{equation}
    }
    \end{itemize}
\end{theorem}
One consequence of this theorem and in particular of this last inequality is that, for all $k>0$, the distance between the particle of index $\lfloor N/2\rfloor+k$ and its neighbor converges towards $\overline{h}$ as $N \to +\infty$. Actually, \eqref{eq:asymptotics} is a quantitative version of this convergence. 

\blue{Early results on one-dimensional crystallization were obtained by Ventevogel and Nijboer~\cite{Ventevogel:On_the_configuration:1978,Ventevogel_Nijboer:On_the_configuration:1979,Ventevogel_Nijboer:On_the_configuration:1979b}. In particular, they studied classes of one-dimensional pair potentials, including Lennard--Jones-type potentials, for which equidistant configurations minimize the energy per particle in the infinite system. Gardner and Radin~\cite{Gardner_Radin:The_infinite-volume:1979} considered a finite one-dimensional Lennard--Jones chain and proved uniqueness of the ground state and convergence toward uniform spacing as $N\to+\infty$. More recently, Jansen, K{\"o}nig, Schmidt and Theil~\cite{Jansen_Konig_Schmidt_Theil_2021} studied one-dimensional Lennard--Jones-type chains at low temperature. Their zero-temperature analysis includes uniqueness of finite-chain minimizers, bulk convergence and a semi-infinite surface-energy problem describing the boundary layer. In the present paper, the one-dimensional interaction comes from the magnetic-colloid model and we obtain explicit finite-$N$ estimates for the spacings. We refer to~\cite{Blanc_Lewin_2015,Canizo_Ramos-Lora:Discrete} for general introductions to crystallization problems and also to~\cite{Blanc_2015,Betermin_Knupfer_Nolte_2020,Radin_Schulman_1983,Radin_1984}.}

\medskip

\blue{Estimate~\eqref{eq:asymptotics} gives quantitative finite-$N$ information that is not contained in the above convergence results. In particular:}
\begin{itemize}
    \item At the center of the spear ($k\simeq N/2$), the estimate~\eqref{eq:asymptotics}
    gives a rate of convergence which is the optimal estimate since it gives exactly that the distances converge to $\overline{h}$ at order $1/N^{\beta-1}$.
    \item
    The farther from the center, the worse the upper bound on the convergence ; but we still have uniform convergence to $\overline{h}$ in the bulk of the spear, away from the extremities: for $k \in [d(N), N-d(N)]$, for any function $d$ such that $d(N) \to +\infty$ as $N\to+\infty$. 
\end{itemize}

Near the extremities ($k\simeq 1$ or $k\simeq N$), the bound~\eqref{eq:asymptotics} is such that it no longer implies convergence. It is a natural question to wonder whether the estimate~\eqref{eq:asymptotics} is optimal and whether the extremities actually converge toward $\overline{h}$. The answer is negative, as evidenced by the following result:
\blue{
\begin{proposition}\label{prop:no convergence}
    Let $H^\ast_N=(h^\ast_{N,k})_{k=1,\dots,N-1}$ be a minimizing gap vector for each $N\geq2$. Then there exists $k_0\in\mathbb N^\ast$ such that
    \begin{equation}
        h^\ast_{N,k_0}\not\longrightarrow\overline{h}
        \qquad\text{as }N\to+\infty.
    \end{equation}
\end{proposition}
Proposition~\ref{prop:no convergence} shows that convergence to $\overline{h}$ does not hold uniformly up to the endpoints. Thus the boundary terms in~\eqref{eq:asymptotics} cannot in general be removed. We do not claim that their precise decay rate is optimal. 
}
The estimate \eqref{eq:asymptotics} is all the worse the closer we get to the extremities, but we cannot conclude to the optimality of this estimate for the intermediate scalings.

The constant $\alpha_*(\beta)$ is defined precisely in paragraph \ref{sec:refined_dist} ; see also Remark \ref{rk:alpha*} for numerics in the physically relevant case $\beta=3$.

\blue{
\subsection{The ring: a circular configuration}
The second configuration considered in this paper is the {\it ring}, which is observed in the numerical simulations of~\cite{Klughertz_2016_these,Godard-Cadillac_Hervieux_Manfredi_2024}; see Figure~\ref{fig:ring}. These simulations also suggest that the ring persists under larger noise levels than the spear. We do not address this stability question here.

Up to translations and rotations, a ring of radius $r>0$ with $N\geq3$ particles is defined by
}

\begin{equation}\label{def:ring}
    x_k=r\begin{pmatrix}
    \cos\left(\frac{2 \pi k}{N}\right)\\
    \sin\left(\frac{2 \pi k}{N}\right)\\
    0
\end{pmatrix}\qquad\text{and}\qquad m_k=\begin{pmatrix}
    -\sin\left(\frac{2\pi k}{N}\right)\\
    \cos \left(\frac{2\pi k}{N}\right)\\
    0
\end{pmatrix}.
\end{equation}
This configuration is denoted $\mR_r\in(\RR^3\times\SS^2)^N$.

We prove a general existence and uniqueness result and compute the optimal radius.

\blue{
\begin{theorem}[Existence and uniqueness of a critical ring structure]\label{thrm:ring}~Consider the magnetic repulsive energy $U$ defined at~\eqref{def:U} with $\alpha>\beta>1$.
\begin{itemize}
 \item[(i)] This energy admits a unique ring-shaped critical point $\mR_r$ for any fixed $N\geq 3$.
    The radius $r := r^\ast_N$ of the ring is given by an explicit formula, see \eqref{def:rN}.
    
\item[(ii)] In the asymptotics $N \to \infty$, recalling the definition of $\overline{h}$ at~\eqref{def:the H gang}, we obtain $\displaystyle r^\ast_N \sim \frac{N}{2 \pi} \overline{h}.$
    Moreover the distance between two nearest particles converges to $\overline{h}$ as $N\to+\infty$ with:
    \begin{equation}\label{Milady de Winter}
        \forall\,k=1,\dots,N,\qquad\quad\Big|\,|x_k-x_{k+1}|\,-\overline{h}\,\Big|\;\lesssim\;\left\{\begin{array}{ll}
            N^{-2} &\quad\text{if }\beta>3,  \\
            N^{-2}\log(N) &\quad\text{if }\beta=3,\\
            N^{1-\beta} &\quad\text{if }1<\beta<3.
        \end{array}\right.
    \end{equation}
    where we used the natural periodicity convention for indices on a circle: $x_{N+1} \coloneqq x_1$.
\end{itemize}
\end{theorem}
}

The study of the ring structure appears to be somewhat less complicated than the spear, since this shape has a lot of symmetry.
Yet the manipulated formulae can remain rather long so we postponed the statement of the formula for $r^\ast_N$ to ease the reading of the main results. 

It is interesting to compare the asymptotics obtained here for the ring structure with the asymptotic for the spear given by Theorem~\ref{thrm:spear}-$(iii)$. 
We observe indeed that the limit distance for the neighboring particles in the ring is the same limit distance $\overline{h}$ that we obtained for the center of the spear.
This is consistent with the property $r^\ast_N\to+\infty$ since, in this asymptotic, the ring shows a local curvature converging to $0$ (the straight line).
In other words, when $N\to+\infty$, the ring and the bulk of the spear have the same asymptotic behavior in the neighborhood of a fixed nanoparticle.

Nevertheless, the rate of convergence to $\overline{h}$ obtained for the ring is lower than that for the spear; we think that this rate may be optimal (from the study of remaining terms in the asymptotics developments, see the proof of Theorem~\ref{thrm:ring} for details).
This difference in the rate of convergence is possibly a consequence of the geometric constraint induced by the curvature of the ring, which implies that the vectors appearing in Formula~\eqref{eq:interaction potential} are not colinear (unlike the case of the spear).

\section{Study of the Spear structure}

\subsection{Reformulation of the problem for the spear}
To study the critical points of $U$ that have the aligned structure called ``spear'', we first compute $\nabla_{x_i} U^d$ and $\nabla_{x_i} U^s$ at the particular point $\mS_P$ defined by~\eqref{def:spear}.  We find:
\begin{align} \label{eq:nabla_xUd(S)}
    \nabla_{x_i} U^d (\mS_P) &=
    \begin{multlined}[t]
        -\sum_{\substack{j=1\\j\neq i}}^{N}\frac{B+B_0}{|r_{ij}|^{\beta+2}}\bigg[(m_i\cdot r_{ij})m_j+(m_j\cdot r_{ij})m_i\\
        +\frac{\beta B_0}{B+B_0}(m_i\cdot m_j)r_{ij}-(\beta+2)\frac{(m_i\cdot r_{ij})(m_j\cdot r_{ij})}{|r_{ij}|^2}r_{ij}\bigg]
    \end{multlined} \notag \\*
        &= \beta B \sum_{\substack{j=1\\j\neq i}}^{N}\frac{(p_i-p_j)}{|p_i-p_j|^{\beta+2}}
        \begin{pmatrix}
            1\\
            0\\
            0
        \end{pmatrix},
\end{align}
%
and
\begin{equation} \label{eq:nabla_xUs(S)}
\nabla_{x_i}U^s (\mS_P) = -\alpha A\sum_{\substack{j=1\\j\neq i}}^{N}\bigg(\frac{1}{|r_{ij}|}\bigg)^{\alpha+1}\frac{r_{ij}}{|r_{ij}|}=-\alpha A \sum_{\substack{j=1\\j\neq i}}^{N}\frac{(p_i-p_j)}{|p_i-p_j|^{\alpha+2}}\begin{pmatrix}
    1\\
    0\\
    0
\end{pmatrix}.
\end{equation}
On the other hand, we also have the following equalities for $\nabla_{m_i} U^d$:
\begin{align}
    \nabla_{m_i} U^d (\mS_P) &= \sum_{\substack{j=1\\j\neq i}}^{N} B_0 \frac{m_j}{|r_{ij}|^\beta} - (B + B_0) \frac{r_{ij} (m_j \cdot r_{ij})}{|r_{ij}|^{\beta+2}} \nonumber \\
        &= \sum_{\substack{j=1\\j\neq i}}^{N} \frac{B_0}{|r_{ij}|^\beta}
        \begin{pmatrix}
            1\\
            0\\
            0
        \end{pmatrix}
        - (B + B_0) \frac{p_i - p_j}{|p_i - p_j|^{\beta+2}}
        \begin{pmatrix}
            p_i - p_j\\
            0\\
            0
        \end{pmatrix} \nonumber \\
        &= - B \sum_{\substack{j=1\\j\neq i}}^{N} \frac{1}{|p_i - p_j|^\beta} 
        \begin{pmatrix}
            1\\
            0\\
            0
        \end{pmatrix}. \label{eq:nabla_mUd(S)}
\end{align}
Such computations yield the following relation between the critical points of $U(\mS_P)$ in $\RR^N$ and the critical points of $U$ in $(\RR^3 \times \SS^2)^N$.

\begin{lemma}
    If $P$  is a critical point of $U(\mS_P)$ in $\RR^N$, then $\mS_P$ is a critical point of $U$ in $(\RR^3 \times \SS^2)^N$.
\end{lemma}

\begin{proof}
    From the expression of $\mS_P$, there obviously holds
    \begin{equation}
        \frac{\partial}{\partial p_i} (U (\mS_P)) = \nabla_{x_i} U (\mS_P) \cdot 
        \begin{pmatrix}
            1\\
            0\\
            0
        \end{pmatrix}.
    \end{equation}
    From the previous computations \eqref{eq:nabla_xUd(S)}-\eqref{eq:nabla_xUs(S)}-\eqref{eq:nabla_mUd(S)}, we know that $\nabla_{x_i} U (\mS_P)$ is collinear to this last vector. Thus, $\frac{\partial}{\partial p_i} (U (\mS_P))$ vanishes if and only if $\nabla_{x_i} U (\mS_P)$ vanishes.
    Moreover, the previous computations also yield that $\nabla_{m_i} U^d (\mS_P)$ is collinear to $m_i$, and therefore so is $\nabla_{m_i} U (\mS_P)$ as $U^s$ does not depend on any $m_i$. Thus, the variation of $U$ in $(\RR^3 \times \SS^2)^N$ with respect to $m_i$ at the point $\mS_P$ is always $0$ in the direction orthogonal to the constraint \blue{$|m_i|=1$}. The conclusion follows from the Lagrange multiplier theorem.
\end{proof}

The resolution of this problem is equivalent to a problem of calculus of variations involving a Lennard-Jones type potential, of the type
\begin{equation}
    L(h):=\frac{A}{|h|^\alpha}-\frac{B}{|h|^\beta}.
\end{equation}
Recall that we have $A,B,\alpha,\beta>0$ with $\alpha>\beta$. 
An example of Lennard-Jones potential is given in Figure~\ref{fig:LJ_potential}.
\begin{figure}[!h]
    \begin{center}
        \includegraphics[width=7cm]{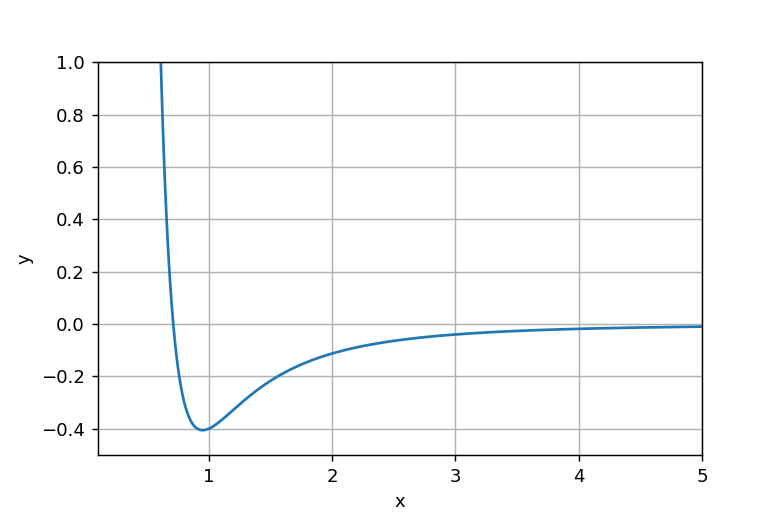}
    \end{center}
    \caption{Typical profile of Lennard-Jones potential}\label{fig:LJ_potential}
\end{figure}
Indeed, consider for $P\in\RR^{N}$ such that $p_i \neq p_j$ for all $i, j$, the function
\begin{equation}\label{def:J L}
    J_0(P):=  U (\mS_P) = \frac{1}{2} \sum_{i=1}^{N}\sum_{\substack{j=1\\j\neq i}}^{N} L\big(|p_i-p_j|\big).
\end{equation}
Then, one can see that finding critical points of $U$ which have the shape of a spear is equivalent to finding a critical point for the function $J_0$.
With this reformulation at hand, the strategy for the following work consists of minimizing the function $J_0$ to obtain the existence of a critical point.
In particular, any critical point solves the system
\begin{equation}\label{eq:Graal}
    \alpha A \sum_{\substack{j=1\\j\neq i}}^{N}\frac{(p_i-p_j)}{|p_i-p_j|^{\alpha+2}}=\beta B \sum_{\substack{j=1\\j\neq i}}^{N}\frac{(p_i-p_j)}{|p_i-p_j|^{\beta+2}}
\end{equation}
\blue{for all $i=1,\dots,N$.}
Since the function $J_0$ is invariant by permutation of particles, we can assume that $p_{k+1} > p_k$ \blue{for all $k=1, \dots, N-1$}. The function then takes this form
\begin{equation}
    J_0(P)=\sum_{i=2}^N\sum_{j=1}^{i-1} L\big(p_i-p_j\big).
\end{equation}
In order to both take into account the invariance by translation, and to have a more pleasant formulation to manipulate, we prefer to work with the distances between particles $h_k:= p_{k+1} - p_{k}>0$ instead of the positions $p_k$. 
It is possible to reformulate the problem only in terms of $h_k$ by noticing that $\sum_{\ell=j}^{i}h_\ell=:p_{i+1}-p_{j}$.
With this setting, the problem of calculus of variation that we study here then consists in minimizing
\begin{equation} \label{def:J L d}
    J(H):=J_0(P)=\sum_{i=1}^{N-1}\sum_{\substack{j=1}}^{i} L\bigg(\sum_{\ell=j}^{i} h_\ell \bigg),\qquad\text{with}\qquad L(h):=\frac{A}{|h|^\alpha}-\frac{B}{|h|^\beta},
\end{equation}
where $H=(h_1,\dots,h_{N-1})\in(\RR_+^\ast)^{N-1}$ is the vector containing all the distances between neighboring nanoparticles. We define
\begin{equation}
    J^\ast:=\inf_{H\in(\RR_+^\ast)^{N-1}}J(H).
\end{equation}
Our problem raises two key questions:
\begin{itemize}
    \item Is $J^\ast$ a minimum of $J$ over $(\RR_+^\ast)^{N-1}$?
    \item Does $J$ have another critical point?
\end{itemize}
We will address both questions in the following sections.

\subsection{First estimates and Existence result}


The function $L$ admits a unique critical point, which is also its global minimizer. Direct computations show that this minimizer is $\widehat{h}$ defined at~\eqref{def:the H gang}.
More precisely, we have $L$ decreasing in $(0, \widehat{h}]$ and increasing in $[\widehat{h}, +\infty)$.
This leads to a first property on the variations of $J$ when one of the components of $H$ is greater than $\widehat{h}$.


\begin{lemma}[Upper bound on the distances]\label{lem:distances upper bound}~
\begin{enumerate}
    \item If $H\in(\RR_+^\ast)^{N-1}$ is a critical point of $J$ then,
    \begin{equation}
        \max_{k=1,\dots,N-1}|h_k|\;\leq\;\widehat{h}.
    \end{equation}
    \item There holds
    \begin{equation}
        J^\ast = \inf_{H \in (0, \widehat{h}]^{N-1}} J(H).
    \end{equation}
\end{enumerate}
\end{lemma}

\begin{proof}
 It suffices to show that if $H\in(\RR_+^\ast)^{N-1}$ such that $h_{k_0} > \widehat{h}$ for some $k_0$, then there holds
    \begin{equation}
        \partial_{h_{k_0}} J (H) > 0.
    \end{equation}

    To start with, we compute the gradient of $J$ at the point $H$. Direct computations give that for all $k=1,\dots, N-1$:
    \begin{equation}\label{eq:grad J}
        \partial_{h_k}J(H)=\sum_{i=1}^k\sum_{j=k}^{N-1}L'\bigg(\sum_{\ell=i}^{j}h_\ell\bigg),\qquad\text{with}\qquad L'(s):=-\frac{\alpha A}{|s|^{\alpha+1}}+\frac{\beta B}{|s|^{\beta+1}}.
    \end{equation}
    We now remark that, since the coordinates of $H$ are all positive:
    \begin{equation}
        \forall\,i\in\{1,\dots,k\},\quad\forall\,j\in\{k,\dots,N-1\},\qquad\sum_{\ell=i}^{j}h_\ell\geq h_{k}.
    \end{equation}
    We use this inequality in the particular case $k=k_0$. Since $h_{k_0}>\widehat{h}$ and $L' > 0$ on $[\widehat{h}, \infty)$, we get:
    \begin{equation}
        \forall\,i\in\{1,\dots,k_0\},\quad \forall\,j\in\{k_0,\dots,N-1\},\qquad L'\bigg(\sum_{\ell=i}^{j}h_\ell\bigg)>0.
    \end{equation}
    Thus,
    \begin{equation}
        \partial_{h_{k_0}}J(H)=\sum_{i=1}^{k_0}\sum_{j={k_0}}^{N-1}L'\bigg(\sum_{\ell=i}^{j}h_\ell\bigg)>0. 
    \end{equation}
    Replacing any coordinate $h_k>\widehat{h}$ by $\widehat{h}$ decreases $J$. Applying this successively to all such coordinates proves (ii).
\end{proof}



With similar arguments, we are able to give a lower bound on the distances between particles for the minimizers.


\begin{lemma}[Lower bound on the distances]\label{lem:distances lower bound}~
    \begin{enumerate}
        \item[(i)] If $H\in(\RR_+^\ast)^{N-1}$ is a critical point of $J$, then
    \begin{equation}
        \min_{i=1,\dots,N-1}|h_i|\;\geq\;\widecheck{h},
    \end{equation}
    where $\widecheck{h}$ is defined in~\eqref{def:the H gang}.
    This is in particular true if $H$ is such that $J(H)=J^\ast$.
        \item[(ii)] \color{blue}The infimum $J^\ast$ is attained, and
    \begin{equation}
        J^\ast = \min_{H \in [\widecheck{h}, \widehat{h}]^{N-1}} J(H).
    \end{equation}\color{black}
        \end{enumerate}
\end{lemma}

\begin{proof} To prove (i), let $H\in(\RR_+^\ast)^{N-1}$ and choose $k_0$ such that $h_{k_0}=\min_k h_k$. We show that $h_{k_0}<\widecheck{h}$ implies $\partial_{h_{k_0}}J(H)<0$. 

    We recall that the gradient of $J$ is computed at \eqref{eq:grad J}. In particular:
    \begin{equation}
        \partial_{h_{k_0}} J (H) = L'\big(h_{k_0}\big) + \sum_{i=1}^{k_0}\sum_{j=k_0+1}^{N-1}L'\bigg(\sum_{\ell=i}^{j}h_\ell\bigg)+\sum_{i=1}^{k_0-1}L'\bigg(\sum_{\ell=i}^{k_0}h_\ell\bigg).
    \end{equation}
    We now use the explicit formula for $L'$ given at~\eqref{eq:grad J} and this implies
    \begin{equation}
        \partial_{h_{k_0}}J(H)\leq \frac{\beta B}{h_{k_0}^{\beta+1}}-\frac{\alpha A}{h_{k_0}^{\alpha+1}}+\beta B\sum_{i=1}^{k_0-1}\left(\sum_{\ell=i}^{k_0}h_\ell\right)^{-(\beta+1)}+\beta B\sum_{i=1}^{k_0}\sum_{j=k_0+1}^{N-1}\left(\sum_{\ell=i}^j h_\ell\right)^{-(\beta+1)}.
    \end{equation}
    We now recall that $h_{k_0}$ is the smallest distance so that we are led to
    \begin{equation}\label{Marais 2}
        \partial_{h_{k_0}} J (H) \leq \frac{\beta B}{h_{k_0}^{\beta+1}} - \frac{\alpha A}{h_{k_0}^{\alpha+1}} + \frac{\beta B}{h_{k_0}^{\beta+1}} \Bigg( \sum_{i=1}^{k_0-1} \frac{1}{|(k_0 + 1)-i|^{\beta+1}} + \sum_{i=1}^{k_0}\sum_{j=k_0+2}^N \frac{1}{|j-i|^{\beta+1}} \Bigg).
    \end{equation}
    Using standard manipulations on the double sums, we write
    \begin{equation}
        \sum_{i=1}^{k_0}\sum_{j=k_0+2}^N\frac{1}{|j-i|^{\beta+1}}= \sum_{i=1}^{k_0} \sum_{k=k_0+2-i}^{N-i} \frac{1}{k^{\beta+1}} = \sum_{k=2}^{N-1} \sum_{i=\max(1, k_0+2-k)}^{\min(k_0, N-k)}\frac{1}{k^{\beta+1}}
    \end{equation}
    Moreover, we can easily prove that
    \begin{equation}
        \min(k_0, N-k) - \max(1, k_0+2-k) + 1 \leq k-1.
    \end{equation}
    Thus,
    \begin{equation}
        \sum_{i=1}^{k_0-1} \frac{1}{|k_0 + 1-i|^{\beta+1}} + \sum_{i=1}^{k_0}\sum_{j=k_0+2}^N\frac{1}{|i-j|^{\beta+1}}\leq \sum_{k=2}^{k_0}\frac{1}{k^{\beta+1}} + \sum_{k=2}^{N-1}\frac{k-1}{k^{\beta+1}} \leq \sum_{k=2}^{+\infty}\frac{1}{k^{\beta}} = \zeta(\beta) - 1
    \end{equation}
    Plugging this back into~\eqref{Marais 2} eventually gives
    \begin{equation}
        \partial_{h_{k_0}} J (H) \leq \frac{\beta B}{h_{k_0}^{\beta+1}} - \frac{\alpha A}{h_{k_0}^{\alpha+1}} + \frac{\beta B}{h_{k_0}^{\beta+1}}(\zeta(\beta)-1) = \frac{\beta B}{h_{k_0}^{\beta+1}} \zeta(\beta) - \frac{\alpha A}{h_{k_0}^{\alpha+1}},
    \end{equation}
    so that, using the fact that $h_{k_0} < \widecheck{h}$ and $\alpha > \beta$,
    \begin{equation}
        \partial_{h_{k_0}} J (H) \leq \frac{1}{h_{k_0}^{\beta+1}} \biggl(\beta B \, \zeta(\beta) - \frac{\alpha A}{h_{k_0}^{\alpha - \beta}} \biggr) < \frac{1}{h_{k_0}^{\beta+1}} \biggl(\beta B \, \zeta(\beta) - \frac{\alpha A}{\widecheck{h}^{\alpha - \beta}} \biggr) = 0,
    \end{equation}
    which proves (i).

  \color{blue}  We now prove (ii). Set $M=N(N-1)/2$, $\ell_m=\min_{s>0}L(s)$ and $\widehat H=(\widehat h,\dots,\widehat h)$. For each $k$,
    \[
    J(H)\geq L(h_k)+(M-1)\ell_m.
    \]
    Since $L(s)\to+\infty$ as $s\to0^+$, there exists $\eta\in(0,\widehat h)$ such that $J(H)>J(\widehat H)$ whenever some $h_k<\eta$. Thus the set
    \[
    K=\left\{H\in[\eta,\widehat{h}]^{N-1}:J(H)\leq J(\widehat H)\right\}
    \]
    is nonempty and compact. Hence $J$ attains its minimum on $K$, and this minimum is global by Lemma~\ref{lem:distances upper bound}. The minimizer $H^*$ is an interior point of $(\RR_+^\ast)^{N-1}$, so $\nabla J(H^*)=0$. Part (i) and Lemma~\ref{lem:distances upper bound} give $H^*\in[\widecheck{h},\widehat{h}]^{N-1}$, proving (ii).
    \color{black}
\end{proof}


The existence of a minimizer for $J$, which is equivalent to the existence of a stationary spear structure for the nano-particle system, is a corollary of these two previous lemmas:
\begin{corollary}\label{lem:existence of minimizer of J}
    The function $J_0$ admits a global minimizer $p^*\in\RR^N$ such that $p_i^*\leq p_j^*$ for $1\leq i\leq j\leq N$, and
    \begin{equation*}
        \widecheck{h}\leq p_{i+1}^*-p_i^*\leq\widehat{h},\qquad i=1,\dots,N-1.
    \end{equation*}
    Moreover, denoting $h_k^* = p_{k+1}^* - p_k^*$, the global minimizer satisfies for all $1 \leq k \leq N-1$
    \begin{equation} \label{eq:crit_point_1d}
        \sum_{i=1}^k \sum_{j=k}^{N-1} L' \Bigg( \sum_{\ell=i}^{j} h_\ell^* \Bigg) = 0.
    \end{equation}
\end{corollary}

\begin{proof}
    \blue{By the proof of Lemma~\ref{lem:distances lower bound}(ii), $J$ has a global minimizer $H^*\in[\widecheck{h},\widehat{h}]^{N-1}$. Set $p_1^*=0$ and $p_{k+1}^*=p_k^*+h_k^*$. Translation and permutation invariance of $J_0$ show that $p^*$ minimizes $J_0$. Equation~\eqref{eq:crit_point_1d} is the condition $\nabla J(H^*)=0$.}
\end{proof}
As a direct consequence of this result, we get the existence of a solution for the problem studied as stated in Theorem~\ref{thrm:spear}-$(i)$. Lemmas \ref{lem:distances upper bound} and \ref{lem:distances lower bound} give the estimates of Theorem~\ref{thrm:spear}-$(ii)$.


\blue{Assume $A=aR^\alpha$, with $a,R>0$ fixed. The following lemma shows that the lower and upper bounds on the neighboring-particle distances both converge to the effective exclusion distance $R$ as $\alpha\to+\infty$. Consequently, the neighboring gaps of any minimizing spear converge uniformly to $R$.}

\begin{lemma}\label{lem:hard sphere}
    Let $a, R > 0$ fixed. Let $\alpha>\beta$ and $A = a R^\alpha$ and define $\widecheck{h}$ and $\widehat{h}$ with~\eqref{def:the H gang}. Then,
    \begin{equation} 
        \widecheck{h},\;\; \widehat{h}\quad\longrightarrow\quad R \quad\qquad\text{as }\alpha\to+\infty.
    \end{equation}
\end{lemma}
\begin{proof}
It is direct to check from the definition of $\widecheck{h}$ and $\widehat{h}$ that they both converge to $R$ as $\alpha\to+\infty$ since the function $t\mapsto t^{\frac{1}{t - \beta}}$ converges to $1$ at $+\infty$.
\end{proof}

\subsection{Convexity property for uniqueness result}

The main purpose of this section is to obtain the uniqueness result stated in Theorem~\ref{thrm:spear}-$(iii)$. 
The main idea of the proof is to establish a strong convexity result on the function $J$ to obtain the uniqueness of the minimizer.

For this, we are naturally led to use convexity of the Lennard Jones potential $L$, and it is convenient to introduce the constants:
\begin{equation}\begin{split}        
    h^\dag &:= \Bigg( \frac{\alpha(\alpha + 1)A}{\beta(\beta + 1)B} \Bigg)^{\frac{1}{\alpha - \beta}} = \widecheck{h}\Bigg( \frac{(\alpha + 1) \zeta (\beta)}{\beta + 1} \Bigg)^\frac{1}{\alpha - \beta}, \\
    h^\ddag &:= \Bigg( \frac{\alpha(\alpha + 1) (\alpha + 2) A}{\beta(\beta + 1) (\beta + 2) B} \Bigg)^\frac{1}{\alpha - \beta} = \widecheck{h} \Bigg( \frac{(\alpha + 1) (\alpha + 2) \, \zeta (\beta)}{(\beta + 1) (\beta + 2)} \Bigg)^\frac{1}{\alpha - \beta}.
\end{split}\end{equation}
\blue{
From the explicit expression of $L''$, we have $L''>0$ on $(0,h^\dag)$ and $L''<0$ on $(h^\dag,+\infty)$. Moreover, $L''$ is decreasing on $(0,h^\ddag)$ and increasing on $(h^\ddag,+\infty)$.

We shall impose $h^\ddag<2\widecheck{h}$: this can be done by choosing $\alpha$ large enough. More precisely, define
\begin{equation}\label{eq:def_alpha_H}
H_\beta(\alpha):=(\alpha-\beta)\log 2-\log\left(\frac{(\alpha+1)(\alpha+2)\zeta(\beta)}{(\beta+1)(\beta+2)}\right).
\end{equation}
We have
\begin{equation}
H_\beta'(\alpha)=\log 2-\frac{1}{\alpha+1}-\frac{1}{\alpha+2},
\qquad
H_\beta''(\alpha)=\frac{1}{(\alpha+1)^2}+\frac{1}{(\alpha+2)^2}>0.
\end{equation}
Together with $H_\beta(\beta^+)=-\log\zeta(\beta)<0$ and $H_\beta(\alpha)\to+\infty$ as $\alpha \to \infty$, this shows that $H_\beta$ has a unique zero $\alpha_H(\beta)>\beta$. Furthermore,
\begin{equation} \label{est:hhh} 
    \alpha>\alpha_H(\beta)\qquad\Longleftrightarrow\qquad h^\ddag<2\widecheck{h}.
\end{equation}
Since $h^\dag<h^\ddag$, we also have $h^\dag<2\widecheck{h}$ for such $\alpha > \alpha_H(\beta)$.
}

\subsubsection{Positivity of the Hessian matrix}\label{sec:cvx}

We are going to prove here that the function $J$ is strongly convex in the admissible set $[\widecheck{h}, \widehat{h}]^{N-1}$ and for that purpose we study the Hessian matrix of $J$.
However, even though we expect uniqueness as soon as $\alpha > \beta$, this study turns out to be much more difficult when $\alpha$ is close to $\beta$, so now we work with large values of $\alpha$ (detailed later). In particular, for a discussion about the values of $\alpha$ which are considered and how it is related to empirical values, see Remark \ref{rk:alpha*}.

We denote the Hessian matrix of $J$ at the point $H=(h_1,\dots,h_{N-1})$ by $\nabla^2J(H)$. The second derivative of $J$ with respect to $h_{\mu}$ and $h_{\nu}$ for the indices $\mu,\nu=1,\dots, N-1$ is denoted $\partial^2_{{h_\mu},{h_\nu}}J(H)$ or more compactly $\partial^2_{\mu\nu}J(H).$ 

Let $\alpha_{\dag}(\beta)$ be defined as the largest zero of $F_{\beta}$, defined by
\begin{equation}\label{eq:def_alpha_dag}
F_{\beta}(\alpha):=\alpha-\beta-\zeta(\beta)^{\frac{\beta+2}{\alpha-\beta}}(\beta+1)(\zeta(\beta+1)+\zeta(\beta)-2).
\end{equation}

Note that $\alpha_{\dag}$ is well defined since $F_{\beta}(\alpha)\to +\infty$ as $\alpha\to+\infty $ and $F_\beta (\alpha) \to - \infty$ as $\alpha \to \beta^+$.

\begin{lemma} \label{lem:Hess_prop}
 \blue{
    Assume that $\alpha>\alpha_H(\beta)$ and let $H \in [\widecheck{h}, \widehat{h}]^{N-1}$. The Hessian of the function $J$ at point $H$ satisfies the following bounds: 
}
    \begin{itemize}
        \item All the diagonal terms satisfy
            \begin{equation} \label{eq:Hess_diag}
                \partial^2_{\mu\mu} J (H) \geq \Lambda_d := \frac{\beta B}{\widehat{h}^{\beta + 2}} (\alpha - \beta) - \frac{\beta (\beta + 1)B}{\widecheck{h}^{\beta + 2}} \big(\zeta (\beta + 1) - 1\big).
            \end{equation}
        \item All the non-diagonal terms are non-positive and satisfy an estimate decreasing with respect to their distance to the diagonal. More precisely:
            \begin{equation} \label{eq:Hess_nodiag}
                0 \geq \partial^2_{\mu\nu} J (H) \geq  - \frac{\Lambda_{nd}}{|\nu-\mu|^{\beta}} \quad \text{where} \quad \Lambda_{nd} :=  \frac{(\beta + 1)B}{\widecheck{h}^{\beta + 2}} >0.
            \end{equation}

        \item The Hessian is a uniformly diagonally dominant matrix: for all $\mu$,
            \begin{equation} \label{eq:Hess_diag_dominant}
                | \partial^2_{\mu\mu} J (H) | - \sum_{\substack{\nu=1\\\nu \neq \mu}}^{N-1} | \partial^2_{\mu\nu} J (H) |
                \geq \Lambda_1,
            \end{equation}
            where
            \begin{equation}
                \Lambda_1\;:=\;
                \frac{\beta B}{\widehat{h}^{\beta + 2}} (\alpha - \beta) - \frac{\beta(\beta + 1)B}{\widecheck{h}^{\beta + 2}} \big(\zeta (\beta + 1)+ \zeta(\beta)-2 \big).
            \end{equation}
    \end{itemize}
    The constant $\alpha_\dag(\beta)$ is defined so that if $\alpha > \alpha_\dag$, then $\Lambda_1, \Lambda_d >0$.
\end{lemma}

\begin{proof}
    From the expression of the functional $J$, we can compute its second partial derivatives:
    \begin{equation} \label{eq:J''}
        \partial^2_{\mu\nu} J (H) = \sum_{i=1}^{\min(\nu, \mu)} \sum_{j=\max(\nu, \mu)+1}^N L'' \bigg( \sum_{\ell=i}^{j-1} h_\ell \bigg).
    \end{equation}
    Now, if $\mu \ne \nu$, all the indices appearing in \eqref{eq:J''} verify $j -i \ge 2$, which yields
    \begin{equation}
        \sum_{\ell=i}^{j-1} h_\ell \geq (j-i) \widecheck{h} \geq 2 \widecheck{h} > h^\dag.
    \end{equation}
    Thus, every term in \eqref{eq:J''} is negative. Moreover, assuming for instance that $\mu > \nu$,
    \begin{equation}
        \partial^2_{\mu\nu} J (H) \geq \sum_{i=1}^{\nu} \sum_{j=\mu+1}^N L'' ((j-i) \widecheck{h})= \sum_{i=1}^{\nu} \sum_{\kappa=\mu+1-i}^{N-i} L'' (\kappa \widecheck{h})=\sum_{\kappa=\mu-\nu+1}^{N-1} \;\sum_{i=\max(1, \mu-\kappa+1)}^{\min(\nu, N - \kappa)} L'' (\kappa \widecheck{h})
    \end{equation}
    Therefore, using that $\min(\nu, N - \kappa) - \max(1, \mu-\kappa+1) + 1 \le \kappa + \nu - \mu$, we get
    \begin{equation}
            \partial^2_{\mu\nu} J (H) \geq \sum_{\kappa=\mu-\nu+1}^{N-1} (\kappa + \nu - \mu) L'' (\kappa \widecheck{h}) \geq \sum_{\kappa=\mu-\nu+1}^{N-1} \kappa L'' (\kappa \widecheck{h}).
    \end{equation}
Using the explicit expression for $L''$ in the above estimate leads to

    \begin{equation}
        \partial^2_{\mu\nu} J (H) \geq - \frac{\beta (\beta + 1) B}{\widecheck{h}^{\beta + 2}} \sum_{\kappa=\mu-\nu+1}^{N-1} \kappa^{-(\beta+1)} \geq - \frac{\beta (\beta + 1) B}{\widecheck{h}^{\beta + 2}} \sum_{\kappa=\mu-\nu+1}^{\infty} \kappa^{-(\beta+1)}.
    \end{equation} 
Moreover, by series-integral comparison, there holds
\begin{equation}
    \sum_{\kappa=\mu-\nu+1}^{\infty} \kappa^{-(\beta+1)} \leq \int_{\mu - \nu}^{+\infty} t^{- (\beta + 1)} \diff t = \frac{(\mu - \nu)^{- \beta}}{\beta}.
\end{equation}
Therefore, we get
\begin{equation}
    \partial^2_{\mu\nu} J (H) \geq - \frac{(\beta + 1)B}{\widecheck{h}^{\beta + 2}} (\mu - \nu)^{- \beta}.
\end{equation}
This gives the conclusion to the second point of the lemma.

We now observe that, for some $\mu$ fixed,
    \begin{equation}\begin{split}
        \sum_{\nu \neq \mu} \partial^2_{\mu\nu} J (H) &= \sum_{\nu = 1}^{\mu - 1} \sum_{i=1}^{\nu} \sum_{j=\mu+1}^N L'' \bigg( \sum_{\ell=i}^{j-1} h_\ell  \bigg) + \sum_{\nu = \mu + 1}^{N-1} \sum_{i=1}^{\mu} \sum_{j=\nu+1}^N L'' \bigg( \sum_{\ell=i}^{j-1} h_\ell  \bigg) \\
            &= \sum_{i=1}^{\mu-1} \sum_{\nu = i}^{\mu - 1} \sum_{j=\mu+1}^N L'' \bigg( \sum_{\ell=i}^{j-1} h_\ell  \bigg) + \sum_{j=\mu+2}^N \sum_{\nu = \mu + 1}^{j-1} \sum_{i=1}^{\mu} L'' \bigg( \sum_{\ell=i}^{j-1} h_\ell  \bigg) \\
            &= \sum_{i=1}^{\mu-1} (\mu - i) \sum_{j=\mu+1}^N L'' \bigg( \sum_{\ell=i}^{j-1} h_\ell  \bigg) + \sum_{j=\mu+2}^N (j - \mu - 1) \sum_{i=1}^{\mu} L'' \bigg( \sum_{\ell=i}^{j-1} h_\ell  \bigg).
    \end{split}\end{equation}
  From this computation, since $\sum_{\ell=i}^{j-1} h_\ell  \geq (j-i) \widecheck{h} \geq 2 \widecheck{h}$ in every term of the two previous double sums, we are in the interval in which $L''$ is increasing and negative in view of \eqref{est:hhh}. Thus:
\begin{equation}
    \sum_{\nu \neq \mu} \partial^2_{\mu\nu} J (H) \geq \sum_{i=1}^{\mu-1} (\mu - i) \sum_{j=\mu+1}^N L'' \bigg( (j-i) \widecheck{h} \bigg) + \sum_{j=\mu+2}^N (j - \mu - 1) \sum_{i=1}^{\mu} L'' \bigg( (j-i) \widecheck{h} \bigg)
\end{equation}
We continue the estimate by relabeling the indices in the sums:
\begin{equation}\begin{split}
        \sum_{\nu \neq \mu} \partial^2_{\mu\nu} J (H)
            &\geq \sum_{\kappa=1}^{\mu-1} \sum_{k=1}^{N-\mu} \kappa L'' \bigg( (k+\kappa) \widecheck{h} \bigg) + \sum_{k=1}^{N-\mu-1} \sum_{\kappa=1}^{\mu} k L'' \bigg( (k+\kappa) \widecheck{h} \bigg) \\
            &\geq \sum_{\kappa=1}^{\mu-1} \sum_{r=1+\kappa}^{N-\mu+\kappa} \kappa L'' ( r \widecheck{h} ) + \sum_{k=1}^{N-\mu-1} \sum_{r=1+k}^{\mu+k} k L'' ( r \widecheck{h} )
    \end{split}
\end{equation}
We now swap the two sums and get
\begin{align}
        \sum_{\nu \neq \mu} \partial^2_{\mu\nu} J (H)
            &\geq \sum_{r=2}^{N-1} \sum_{\kappa=\max(1, r + \mu - N)}^{\min(\mu-1, r-1)} \kappa L'' ( r \widecheck{h} ) + \sum_{r=2}^{N-1} \sum_{k=\max(1, r - \mu)}^{\min(N-\mu-1, r-1)} k L'' ( r \widecheck{h} ) \nonumber \\
            &\geq \sum_{r=2}^{N-1} \frac{r (r-1)}{2} L'' ( r \widecheck{h} ) + \sum_{r=2}^{N-1} \frac{r (r-1)}{2} L'' ( r \widecheck{h} ) \geq \sum_{r=2}^{N-1} r^2 L'' ( r \widecheck{h} ) \nonumber \\
            &\geq - \frac{\beta(\beta + 1)B}{\widecheck{h}^{\beta + 2}} \big(\zeta(\beta) - 1\big). \label{et la bete}
\end{align}
In the last inequality, we used the explicit formula for $L''$ and the definition of the zeta function $\zeta$. The diagonal terms of the Hessian can be estimated in a similar fashion. To start with, we compute:
    \begin{equation}\begin{split}
        \partial^2_{\mu\mu} J (H) &= \sum_{i=1}^{\mu} \sum_{j=\mu+1}^N L'' \bigg( \sum_{\ell=i}^{j-1} h_\ell \bigg) = L'' (h_\mu) + \sum_{i=1}^{\mu} \sum_{j=\max(i+2, \mu+1)}^N L'' \bigg( \sum_{\ell=i}^{j-1} h_\ell \bigg).
    \end{split}\end{equation}
    Since $h_\mu \in [\widecheck{h}, \widehat{h}]$ and $h^\dag > \widehat{h}$ and $h^\ddag > \widehat{h}$, the first term is easy to estimate: \begin{equation}L'' (h_\mu) \geq L'' (\widehat{h}) = \frac{\beta B}{\widehat{h}^{\beta + 2}} (\alpha - \beta).\end{equation}
    As for the other terms, they are all negative and can be estimated following the same ideas leading to \eqref{et la bete}:
    \begin{align}
      \MoveEqLeft  \sum_{i=1}^{\mu} \sum_{j=\max(i+2, \mu+1)}^N L'' \bigg( \sum_{\ell=i}^{j-1} h_\ell \bigg) \geq \sum_{i=1}^{\mu} \sum_{j=\max(i+2, \mu+1)}^N L'' \bigg( (j-i) \widecheck{h} \bigg) \\
            &\geq \sum_{i=1}^{\mu} \sum_{\kappa=\max(2, \mu-i+1)}^{N-i} L'' ( \kappa \widecheck{h} ) \geq \sum_{\kappa=2}^{N-1} \sum_{i=\max(1, \mu - \kappa + 1)}^{\min(N-\kappa, \mu)} L'' ( \kappa \widecheck{h} ) \geq \sum_{\kappa=2}^{N-1} \kappa L'' ( \kappa \widecheck{h} ) \\
            &\geq - \frac{\beta (\beta + 1)B}{\widecheck{h}^{\beta + 2}} \big(\zeta (\beta + 1) - 1\big).
   \end{align}
    This gives us the first point of the lemma.\medskip

    The last point to prove is the diagonally dominant property \eqref{eq:Hess_diag_dominant}. It is obtained directly from the above estimate (for the diagonal term) combined with~\eqref{et la bete} (for the non-diagonal terms).
\end{proof}

\subsubsection{Uniqueness result}

\blue{
The uniqueness of the minimizer for $J$,  can now be proved using the strong convexity property of the functional $J$ induced by Lemma~\ref{lem:Hess_prop}.

\begin{lemma}\label{lem:uniqueness}
    If $\alpha>\max\{\alpha_\dag(\beta),\alpha_H(\beta)\}$, the function $H\mapsto J (H)$ is strongly convex in $[\widecheck{h}, \widehat{h}]^{N-1}$. Furthermore, $J$ admits a unique critical point which is its global minimizer.
\end{lemma}
}
\begin{proof}
    From Lemma \ref{lem:Hess_prop} and the well-known Gershgorin circle theorem, it is direct to show that the lowest eigenvalue of $\nabla^2 J(H)$ is larger than $\Lambda_1$ for all $H\in[\widecheck{h}, \widehat{h}]^{N-1}$. We recall that for $\alpha>\alpha_\dag$ we have $\Lambda_1$ that is a positive constant. Thus $J$ is strongly convex on this subset. On the other hand, Lemmas \ref{lem:distances upper bound} and \ref{lem:distances lower bound} show that any critical point has to be in this subset, which proves the uniqueness of the critical point, and then also of the minimizer.
\end{proof}

\subsection{Asymptotic limit on the distances}


\subsubsection{A refined Gershgorin estimate}
Before we go further toward the proof of the convergence result claimed by Theorem~\ref{thrm:spear}-$(iii)$, we state here a refinement of the classical Gershgorin circle theorem: more precisely, we quantify how the (polynomial) decay of coefficients of a matrix away from the diagonal implies a similar decay for the coefficients of its inverse matrix.

\begin{proposition}[Quantitative Gershgorin estimate] \label{prop:matrix_inverse}
Let $\gamma>1$ and $\delta = 2(1+2^\gamma) \zeta(\gamma)$.  Let $c>0$ and $d>0$ be such that
\begin{equation}\label{def:r+}
    r_+:= \frac{c}{d} \frac{\delta + \sqrt{\delta^2 + 8 \zeta (2 \gamma)}}{2} < 1,
\end{equation}
where $\zeta$ is the standard zeta function.
Let $N \in \m N^*$ and  $A \in \mathcal M_N(\mathbb C)$ such that for all $i \ne j$,
\begin{equation} 
|A_{ij}| \le \frac{c}{|i-j|^\gamma} \quad \text{and} \quad  |A_{ii}| \ge d. \end{equation}

\blue{
Then $A$ is invertible and there exists $\kappa = \kappa(\gamma, \frac{c}{d})$ (bounded as $\frac{c}{d} \to 0$) such that for all $i \ne j$,
\begin{equation} |(A^{-1})_{ij}| \le \kappa  \frac{c}{d^2 |i-j|^\gamma} \quad \text{and} \quad  |(A^{-1})_{ii}| \le \kappa d^{-1}. \end{equation}
}
\end{proposition}

The main interest of this proposition for us is to combine this with the convexity property established in Lemma~\ref{lem:Hess_prop}. 

\blue{
We shall apply Proposition~\ref{prop:matrix_inverse} to the Hessian in the proof of Theorem~\ref{thrm:spear}-$(iii)$.

Jaffard~\cite{Jaffard_1990} proved that, for suitable infinite matrices with polynomial or exponential off-diagonal decay, bounded invertibility on $\ell^2$ implies the same type of decay for the inverse. Proposition~\ref{prop:matrix_inverse} is a finite-dimensional estimate: invertibility follows from the explicit condition~\eqref{def:r+}, and the constants are uniform in $N$ and explicit in terms of $\gamma$ and $c/d$. We include the proof because these constants are used below.
}

\begin{remark}
We observe that the hypothesis on the number $r_+$ in Proposition~\ref{prop:matrix_inverse} implies that the matrix $A$ is strictly diagonally dominant: 
\begin{equation}
\forall\,i=1,\dots,N, \qquad \sum_{\substack{j=1\\j\neq i}}^N \big|A_{ij}|\,<\,\big|A_{ii}\big|.\end{equation}
To this extent, the estimate stated by Proposition~\ref{prop:matrix_inverse} can indeed be seen as a refinement of the well-known Gershgorin circle theorem which ensures invertibility and control on the eigenvalues. Proposition~\ref{prop:matrix_inverse} provides extra information on the decay of the coefficients away from the diagonal for $A^{-1}$ and the proof of this property is a natural generalization and refinement of Gershgorin's arguments.
\end{remark}

\subsubsection{No uniform convergence for the distances: proof of Proposition~\ref{prop:no convergence}}
The next objective of this article is to study the asymptotic of the spear structure in the limit where the number of nanoparticles reaches $+\infty$.
Near the center of the spear, we expect the distances between the particles to converge to some value $\overline{h}$. 
Such a value must then satisfy an asymptotic equation based on \eqref{eq:crit_point_1d}. More precisely, if we do a re-indexation of the sums in this equation, we get
\begin{equation}
\sum_{i=1-k}^0\sum_{j=1}^{N-k}L'\bigg(\sum_{\ell=i+k}^{j+k-1}h_{ \ell}^\ast\bigg)\;=\;0, \qquad \text{for all } k = 1, \dots, N-1.
\end{equation}
Recall that $h_{\ell}^\ast$ has been defined in Corollary \ref{lem:existence of minimizer of J}.
If we now insert in the equality above the ansatz that all the distances are equal to the same distance $\overline{h}$ when $N\to+\infty$ in the center of the spear (meaning $k\simeq N/2$), we get formally the following asymptotic equation:
\begin{equation} \label{eq:h_c 1}
    \sum_{i=-\infty}^0 \sum_{j=1}^\infty L' \Big( (j-i)\overline{h} \Big) = 0.
\end{equation}
By direct computations (recall that $L' (x) = - \frac{\alpha A}{x^{\alpha+1}} + \frac{\beta B}{x^{\beta + 1}}$ for $x > 0$), there hold
\begin{align}
    \sum_{i=-\infty}^0 \sum_{j=1}^\infty L' \Big( (j-i) \overline{h} \Big) &= \sum_{j=1}^\infty \sum_{\kappa=j}^\infty L' \big( \kappa \overline{h} \big) = \sum_{\kappa=1}^\infty \sum_{j=1}^\kappa L' \big( \kappa \overline{h} \big)= \sum_{\kappa=1}^\infty \kappa L' \big( \kappa \overline{h} \big) \\*
        &= - \frac{\alpha A}{\overline{h}^{\alpha+1}} \zeta (\alpha) + \frac{\beta B}{\overline{h}^{\beta + 1}} \zeta (\beta).
\end{align}
Therefore, \eqref{eq:h_c 1} is equivalent to
\begin{equation} \label{def:overline h}
    \overline{h}:= \bigg( \frac{\alpha A \zeta (\alpha)}{\beta B\zeta (\beta)} \bigg)^\frac{1}{\alpha-\beta}.
\end{equation}
These manipulations with this ansatz explain the emergence of this distance $\overline{h}$ as a candidate for the limit distances (at least in the center of the structure).

There is also another natural distance $\widetilde{h}$ obtained from the same considerations and ansatz but for the boundary of the spear (the very last or first distance). Indeed, if we go back to~\eqref{eq:crit_point_1d} with $k=1$ and let $N\to+\infty$ with the ansatz that all the distance converge with the same distance $\widetilde{h}$, we obtain the following equation:
\begin{equation} 
    \sum_{j=1}^\infty L' \big( j \widetilde{h} \big) = 0.
\end{equation}
Similar computations, using again the explicit expression of $L'$, show that the equation above admits a unique solution:
\begin{equation} \label{def:h_b}
    \widetilde{h}:= \bigg( \frac{\alpha A \zeta (\alpha + 1)}{\beta B \zeta (\beta + 1)} \bigg)^\frac{1}{\alpha-\beta}.
\end{equation}
According to Lemma~\ref{lem:widecheck} in the appendix, we have $\overline{h}<\widetilde{h}<\widehat{h}$. As a consequence we can prove that Proposition~\ref{prop:no convergence} follows from these computations above and Theorem~\ref{thrm:spear}:\medskip

\begin{proof}[Proof of Proposition~\ref{prop:no convergence}.]
\blue{
In this proof, we write $H^\ast_N=(h^\ast_{N,\ell})_{\ell=1,\dots,N-1}$ for the minimizing gap vector of the $N$-particle system.
}

By contradiction, assume that for every fixed $k\in\mathbb N^\ast$,
\begin{equation}
    h^\ast_{N,k}\longrightarrow\overline{h}
    \qquad\text{as }N\to+\infty.
\end{equation}
Then, for every fixed $M\geq2$,
\begin{equation}\label{def:epsilon N}
    \varepsilon_{M,N}:=\max_{k=1,\dots,M-1}\big|h^\ast_{N,k}-\overline{h}\big|\longrightarrow0
    \qquad\text{as }N\to+\infty.
\end{equation} For any fixed value of $N$ with $k=1$, Equation~\eqref{eq:crit_point_1d} writes
\begin{equation} \label{eq:yet another critical point relation} \sum_{j=1}^{+\infty}L'\bigg(\sum_{\ell=1}^jh^\ast_{N,\ell}\bigg)\mathbbm{1}_{j\leq N-1}\;=\;0.\end{equation}
Using now the mean value theorem with Lemmas \ref{lem:distances upper bound} and \ref{lem:distances lower bound}, for any $M < N$, we have
\begin{equation}\begin{split}
    \bigg|\sum_{j=1}^{+\infty}L'\bigg(\sum_{\ell=1}^jh^\ast_{N,\ell}\bigg)&\mathbbm{1}_{j\leq N-1}-\sum_{j=1}^{+\infty}L'\big(j\overline{h}\big)\bigg|\;\leq\;\sum_{j=1}^{N-1}\bigg|L'\bigg(\sum_{\ell=1}^jh^\ast_{N,\ell}\bigg)-L'\big(j\overline{h}\big)\bigg|+\sum_{j=N}^{+\infty}\Big|L'\big(j\overline{h}\big)\Big|\\
    &\leq\;\sum_{j=1}^{N-1} \bigg| \sum_{\ell=1}^j h^\ast_{N,\ell} - j \overline{h} \bigg| \sup_{\zeta \in[j \widecheck{h},j \widehat{h}]}\;\big|L''(\zeta)\big|+\sum_{j=N}^{+\infty}\Big|L'\big(j\overline{h}\big)\Big|\\
    &\leq M\varepsilon_{M,N}\sum_{j=1}^{M-1}\sup_{\zeta \in[j \widecheck{h},j \widehat{h}]}\;\big|L''(\zeta)\big| + \sum_{j=M}^{N-1} j (\widehat{h} - \widecheck{h}) \sup_{\zeta \in[j \widecheck{h},j \widehat{h}]}\;\big|L''(\zeta)\big| +\sum_{j=N}^{+\infty}\Big|L'\big(j\overline{h}\big)\Big|,
\end{split}
\end{equation}
where $\varepsilon_{M,N}$ was defined at~\eqref{def:epsilon N}.
If $j_0$ is large enough so that $j_0 \widecheck{h} \geq h^\ddag \geq h^\dag$, then $L''$ is negative and increasing on $[j \widecheck{h},j \widehat{h}]$ for any $j \geq j_0$, so that
\begin{equation}
    \sup_{\zeta \in[j \widecheck{h},j \widehat{h}]}\;\big|L''(\zeta)\big| \leq - L''(j \widecheck{h}) \leq \frac{B \beta (\beta + 1)}{(j \widecheck{h})^{\beta + 2}}.
\end{equation}
Similarly, we get for any $j > j_0$ (up to taking a greater $j_0$)
\begin{equation}
    \Big|L'\big(j\overline{h}\big)\Big| \leq \frac{B \beta}{(j \overline{h})^{\beta + 1}}.
\end{equation}
Thus, if $N > M > j_0$,
\begin{multline*}
\bigg|\sum_{j=1}^{+\infty}L'\bigg(\sum_{\ell=1}^jh^\ast_{N,\ell}\bigg)\mathbbm{1}_{j\leq N-1}-\sum_{j=1}^{+\infty}L'\big(j\overline{h}\big)\bigg| \\
    \begin{aligned}
        &\leq M\varepsilon_{M,N} \Bigl(\sum_{j=1}^{j_0}\sup_{\zeta \in[j \widecheck{h},j \widehat{h}]}\;\big|L''(\zeta)\big| + \sum_{j=j_0+1}^{M-1}\frac{B \beta (\beta + 1)}{(j \widecheck{h})^{\beta + 2}} \Bigr) + \sum_{j=M}^{N-1} j (\widehat{h} - \widecheck{h}) \frac{B \beta (\beta + 1)}{(j \widecheck{h})^{\beta + 2}} +\sum_{j=N}^{+\infty}\frac{B \beta}{(j \overline{h})^{\beta + 1}} \\
        &\leq M\varepsilon_{M,N} \Bigl(C + \frac{B \beta (\beta + 1)}{\widecheck{h}^{\beta + 2}} \sum_{j=j_0+1}^{M-1} j^{-(\beta + 2)}\Bigr) + (\widehat{h} - \widecheck{h}) \frac{B \beta (\beta + 1)}{\widecheck{h}^{\beta + 2}} \sum_{j=M}^{\infty} j^{-(\beta + 1)} + \frac{B \beta}{\overline{h}^{\beta + 1}} \sum_{j=N}^{+\infty} j^{- (\beta + 1)}.
    \end{aligned}
\end{multline*}

Let $\varepsilon > 0$, and fix $M > j_0$ such that
\begin{equation}
    (\widehat{h} - \widecheck{h}) \frac{B \beta (\beta + 1)}{\widecheck{h}^{\beta + 2}} \sum_{j=M}^{\infty} j^{-(\beta + 1)} < \varepsilon,
\end{equation}
which is possible since $\beta + 1 > 2$. Then, by taking the limit $N \to + \infty$, we get
\begin{equation}
    \limsup_{N \to + \infty} \bigg|\sum_{j=1}^{+\infty}L'\bigg(\sum_{\ell=1}^jh^\ast_{N,\ell}\bigg)\mathbbm{1}_{j\leq N-1}-\sum_{j=1}^{+\infty}L'\big(j\overline{h}\big)\bigg| \leq \varepsilon.
\end{equation}
This is true for all $\varepsilon > 0$, so that we get (thanks to \eqref{eq:yet another critical point relation})
\begin{equation} \sum_{j=1}^{+\infty}L'\big(j\overline{h}\big)=0. \end{equation}
Straightforward computations give that the only solution to the equation above is $\widetilde{h}$ defined by~\eqref{def:h_b}. 
Nevertheless, we prove in Lemma~\ref{lem:widecheck} in the appendix that $\overline{h}\neq\widetilde{h}$, which is a contradiction.
\end{proof}


\subsubsection{A first refined estimate on the distances} \label{sec:refined_dist}

From the ansatz presented just before, it is possible to do more precise computations on the minimizer and improve Lemmas \ref{lem:distances upper bound} and \ref{lem:distances lower bound}:

\blue{
Let $\alpha_G(\beta)$ be defined as the largest zero of the function $G_{\beta}$ 
\begin{equation}\label{eq:def_alpha_ast}
G_{\beta}(\alpha):=\frac{2}{\delta+\sqrt{\delta^2+8\zeta(2\beta)}}\left(\frac{\beta(\alpha-\beta)}{(\beta+1)\zeta(\beta)^{\frac{\beta+2}{\alpha-\beta}}}-\beta(\zeta(\beta+1)-1)\right)-1,
\end{equation}
where $\delta = 2(1+2^\beta) \zeta(\beta)$ appears in Proposition~\ref{prop:matrix_inverse}. The condition $\alpha>\alpha_G(\beta)$ is exactly the condition~\eqref{def:r+} obtained when Proposition~\ref{prop:matrix_inverse} is applied with $\gamma=\beta$, $c=\Lambda_{nd}$ and $d=\Lambda_d$. Note that $\alpha_G(\beta)$ is well defined since $G_{\beta}(\alpha)\to+\infty$ as $\alpha\to+\infty$ and $G_\beta (\alpha) \to - \frac{2 \beta(\zeta(\beta+1)-1)}{\delta+\sqrt{\delta^2+8\zeta(2\beta)}} - 1 < 0$ as $\alpha \to \beta^+$. Moreover, $\alpha_G(\beta)>\alpha_\dag(\beta)$, see Lemma~\ref{lem:a*>ad} in the appendix.

We also set
\begin{equation}\label{eq:def_alpha_Z}
\alpha_Z(\beta):=(\beta+1)\zeta(\beta)-1,
\end{equation}
and define the threshold used throughout the remainder of the proof by
\begin{equation}\label{eq:def_alpha_master}
\alpha_\ast(\beta):=\max\{\alpha_G(\beta),\alpha_H(\beta),\alpha_Z(\beta)\}.
\end{equation}
If $\alpha>\alpha_\ast(\beta)$, then $h^\ddag<2\widecheck{h}$ by~\eqref{est:hhh}, while $\alpha>\alpha_Z(\beta)$ implies
\begin{equation}\label{eq:alpha_Z_consequences}
(\alpha+1)\zeta(\alpha+1)>(\beta+1)\zeta(\beta+1),\qquad (\alpha+1)\zeta(\alpha)>(\beta+1)\zeta(\beta).
\end{equation}
Indeed, for the former inequality, $\alpha+1>(\beta+1)\zeta(\beta)$ and $\zeta(s)>1$; and for the latter, $\zeta$ is decreasing on $(1,+\infty)$.
Thus all the restrictions on $\alpha$ used in the proof below are contained in the single assumption $\alpha>\alpha_\ast(\beta)$.

\begin{lemma} \label{lem:better_up_low_bound}
    Let $H^* \in (\RR_+^\ast)^{N-1}$ be a minimizer for the function $J$ with $\alpha>\beta > 1$. Assume moreover that $\alpha>\alpha_\ast(\beta)$. Then there exist two constants $c,C>0$, depending on $A,B,\alpha,\beta$ but independent of $N$, such that for all $1 \leq k \leq N-1$,
    \begin{equation}
        \overline{h} + \frac{c}{N^{\beta - 1}} \leq h_k^\ast \leq \widetilde{h} + \frac{C}{N^\beta},
    \end{equation}
    where $\overline{h}$ (resp. $\widetilde{h}$) are defined in \eqref{def:overline h} (resp. in \eqref{def:h_b}). 
\end{lemma}
}

\begin{proof}
\blue{
    First, since $\alpha>\alpha_\ast(\beta)\geq\alpha_H(\beta)$, relation~\eqref{est:hhh} gives $h^\dag<2\widecheck{h}$. Since $L''<0$ on $[h^\dag,\infty)$, the function $L'$ is therefore decreasing on $[2\widecheck{h},\infty)$.
    
    Let $i_\textnormal{max}$ such that $h_{i_\textnormal{max}}^* = \max_k h_k^* =: {h_\textnormal{max}}$. Then, from \eqref{eq:crit_point_1d} at $k = i_\textnormal{max}$, we get
    \begin{equation} \label{eq:crit_point_max}
        - L' ({h_\textnormal{max}}) = \sum_{i = 1}^{i_\textnormal{max}} \sum_{j =\max(i_\textnormal{max} + 1, i+2)}^N L' \bigg( \sum_{\ell=i}^{j-1} h_\ell^* \bigg).
    \end{equation}
    }
    In the double sum in the right-hand side, we now have $j-i \geq 2$. Thus, there holds
    \begin{equation}
        (j-i) {h_\textnormal{max}} = \sum_{\ell=i}^{j-1} {h_\textnormal{max}} \geq \sum_{\ell=i}^{j-1} h_\ell^* \geq \sum_{\ell=i}^{j-1} \widecheck{h} = (j-i) \widecheck{h} \geq 2 \widecheck{h}.
    \end{equation}
    Using the fact that $L'$ is decreasing for such values, we get from \eqref{eq:crit_point_max}
    \begin{equation}
        - L' ({h_\textnormal{max}}) \geq \sum_{i = 1}^{i_\textnormal{max}} \sum_{j =\max(i_\textnormal{max} + 1, i+2)}^N  L' \bigg( (j-i) {h_\textnormal{max}} \bigg),
    \end{equation}
    which can be rewritten as
    \begin{equation} \label{eq:hmax_prop 1}
        \sum_{i = 1}^{i_\textnormal{max}}  \sum_{j =i_\textnormal{max} + 1}^N L' \bigg( (j-i) {h_\textnormal{max}} \bigg) \leq 0.
    \end{equation}
    On the other hand, we can compute
    \begin{equation}\begin{split}
        \sum_{i = 1}^{i_\textnormal{max}} \sum_{j = i_\textnormal{max} + 1}^N & L' \bigg( (j-i) {h_\textnormal{max}} \bigg) = \sum_{i = 1}^{i_\textnormal{max}} \sum_{\ell = i_\textnormal{max} + 1 - i}^{N-i} L' ( \ell {h_\textnormal{max}} ) =  \sum_{\ell = 1}^{N-1} \sum_{i = \max(1, i_\textnormal{max} + 1 - \ell)}^{\min(i_\textnormal{max}, N-\ell)} L' ( \ell {h_\textnormal{max}} ) \\
            &= \sum_{\ell = 1}^{N-1} \bigg(\min(i_\textnormal{max}, N-\ell) - \max(1, i_\textnormal{max} + 1 - \ell) + 1 \bigg) L' ( \ell {h_\textnormal{max}} ) \\
            &= L' ({h_\textnormal{max}}) + \sum_{\ell = 2}^{N-1} \bigg(\min(i_\textnormal{max}, N-\ell) + \min(0, \ell - i_\textnormal{max}) \bigg) L' ( \ell {h_\textnormal{max}} ).
    \end{split}\end{equation}
    For $1\leq i_\textnormal{max}\leq N-1$ and $1\leq\ell\leq N-1$, the summation coefficient is at least $1$ (with equality if $i_\textnormal{max}=1$ or $i_\textnormal{max}=N-1$). Moreover, for $\ell\geq2$, we have $\ell h_\textnormal{max}\geq2\widecheck{h}>h^\dag>\widehat{h}$ by~\eqref{est:hhh}, hence $L'(\ell h_\textnormal{max})\geq0$. Thus, for all $N\geq2$,
    \begin{equation}
        \sum_{i = 1}^{i_\textnormal{max}} \sum_{j = i_\textnormal{max} + 1}^N L' \bigg( (j-i) {h_\textnormal{max}} \bigg) \geq L' ({h_\textnormal{max}}) + \sum_{\ell = 2}^{N-1} L' ( \ell {h_\textnormal{max}} ) = \sum_{\ell = 1}^{N-1} L' ( \ell {h_\textnormal{max}} ).
    \end{equation}
    Plugging this into \eqref{eq:hmax_prop 1} gives
    \begin{equation}
        \sum_{\ell = 1}^{N-1} L' ( \ell {h_\textnormal{max}} ) \leq 0,
    \end{equation}
    or also
    \begin{equation} \label{eq:fmax_hmax}
        {L^\sharp} ({h_\textnormal{max}}) \leq \sum_{\ell = N}^\infty L'(\ell {h_\textnormal{max}}),
    \end{equation}
    where
    \begin{equation}
        {L^\sharp} (x):= \sum_{\ell = 1}^\infty L'(\ell x) = \frac{\beta B}{x^{\beta + 1}} \zeta (\beta + 1) - \frac{\alpha A}{x^{\alpha+1}} \zeta (\alpha + 1).
    \end{equation}
 \blue{
   From direct computations, we know that ${L^\sharp}$ is increasing on $(0, {h^\sharp})$ where
    \begin{equation} \label{eq:h_sharp}
        {h^\sharp}:= \Bigg( \frac{\alpha(\alpha+1)A \, \zeta (\alpha + 1)}{\beta(\beta+1)B \, \zeta (\beta + 1)} \Bigg)^\frac{1}{\alpha - \beta} = \widehat{h} \, \Bigg( \frac{(\alpha + 1) \, \zeta (\alpha + 1)}{(\beta + 1) \, \zeta (\beta + 1)} \Bigg)^\frac{1}{\alpha - \beta} > \widehat{h},
    \end{equation}
    where the last inequality follows from the first inequality in~\eqref{eq:alpha_Z_consequences}.
    }
    
    Moreover, by definition, we know that ${L^\sharp} (\widetilde{h}) = 0$.
    On the other hand, it holds
    \begin{equation} \label{eq:hmax reminder 1}
        \sum_{\ell = N}^\infty L'(\ell {h_\textnormal{max}}) \leq \sum_{\ell = N}^\infty \frac{\beta B}{(\ell {h_\textnormal{max}})^{\beta + 1}} \leq \frac{\beta B}{h_\textnormal{max}^{\beta + 1}} \sum_{\ell = N}^\infty \ell^{-(\beta + 1)} \leq \frac{B}{\widecheck{h}^{\beta + 1} (N-1)^\beta}.
    \end{equation}  
\blue{
    If $h_\textnormal{max}\leq\widetilde{h}$, the desired upper bound is immediate. We may therefore assume $h_\textnormal{max}>\widetilde{h}$. We denote by ${L^\sharp}^{-1}$ the inverse of the restriction of $L^\sharp$ to $[\widetilde{h},\widehat{h}]$. Since $h_\textnormal{max}\leq\widehat{h}$ and $L^\sharp$ is strictly increasing on this interval, there exists $N_0$, depending only on $A,B,\alpha,\beta$, such that for every $N\geq N_0$,
    \begin{equation}
        \frac{B}{\widecheck{h}^{\beta + 1} (N-1)^\beta} \leq L^\sharp(\widehat{h}).
    \end{equation}
    Hence, by \eqref{eq:fmax_hmax} and \eqref{eq:hmax reminder 1},
    \begin{equation}
        {h_\textnormal{max}} \leq {L^\sharp}^{-1} \bigg( \frac{B}{\widecheck{h}^{\beta + 1} (N-1)^\beta} \bigg)
    \end{equation}
    for every $N\geq N_0$. Therefore,
    \begin{equation}\label{Zeus}
        {h_\textnormal{max}} - \widetilde{h} \leq \frac{1}{\min_{[\widetilde{h}, \widehat{h}]} {L^\sharp}'} \frac{B}{\widecheck{h}^{\beta + 1} (N-1)^\beta}.
    \end{equation}
    Since $h_\textnormal{max}\leq\widehat{h}$, enlarging the constant $C$ also gives the desired estimate for the finitely many values $2\leq N<N_0$.
}
    This leads to the expected upper bound. 
    
    As for the lower bound, by similar computations with $i_\textnormal{min}$ such that $h_{i_\textnormal{min}}^* = \min_k h_k^* =: h_\textnormal{min}$, we get
    \begin{equation} 
        L' (h_\textnormal{min}) + \sum_{\ell = 2}^{N-1} \bigg(\min(i_\textnormal{min}, N-\ell) + \min(0, \ell - i_\textnormal{min}) \bigg) L' ( \ell h_\textnormal{min} ) \geq 0.
    \end{equation}
    Moreover, $\min(i_\textnormal{min}, N-\ell) + \min(0, \ell - i_\textnormal{min}) \leq \min(\ell, N - \ell)$, and this bound is reached for $i_\textnormal{min} = \lfloor \frac{N}{2} \rfloor$. Thus, we get
    \begin{equation}
        \sum_{\ell = 1}^{\lfloor \frac{N}{2} \rfloor} \ell L' ( \ell h_\textnormal{min} ) + \sum_{\ell = \lfloor \frac{N}{2} \rfloor + 1}^{N} (N - \ell) L' ( \ell h_\textnormal{min} ) \geq 0,
    \end{equation}
    which can be rewritten as
    \begin{equation} \label{eq:fmin_hmin}
        {L^\flat} (h_\textnormal{min}) \geq \sum_{\ell = \lfloor \frac{N}{2} \rfloor + 1}^{\infty} \ell L' ( \ell h_\textnormal{min} ) - \sum_{\ell = \lfloor \frac{N}{2} \rfloor + 1}^{N} (N - \ell) L' ( \ell h_\textnormal{min} ), 
    \end{equation}
    where
    \begin{equation}
        {L^\flat} (x):= \sum_{\ell = 1}^{\infty} \ell L' ( \ell x ) = \frac{\beta B}{x^{\beta + 1}} \zeta (\beta) - \frac{\alpha A}{x^{\alpha+1}} \zeta (\alpha).
    \end{equation}
    Similarly as before, we have ${L^\flat}' > 0$ on $(0, {h^\flat})$ where
 \blue{
    \begin{equation} \label{eq:h_flat}
        {h^\flat}:= \Bigg( \frac{\alpha(\alpha+1)A \, \zeta (\alpha)}{\beta(\beta+1)B \, \zeta (\beta)} \Bigg)^\frac{1}{\alpha - \beta} = \Bigg( \frac{(\alpha+1)\, \zeta (\alpha)}{(\beta+1)\, \zeta (\beta)} \Bigg)^\frac{1}{\alpha - \beta}\widehat{h}\;>\;\widehat{h},
    \end{equation}
    where the last inequality follows from the second inequality in~\eqref{eq:alpha_Z_consequences}.
    Writing $m=\lfloor N/2\rfloor$, the right-hand side of~\eqref{eq:fmin_hmin} equals
    \[
    \sum_{\ell=m+1}^{N}(2\ell-N)L'(\ell h_\textnormal{min})+\sum_{\ell=N+1}^{\infty}\ell L'(\ell h_\textnormal{min}).
    \]
    The first sum is nonnegative, since $\ell h_\textnormal{min}\geq2\widecheck{h}>\widehat{h}$. For $\ell\geq N+1\geq3$, we have $\ell h_\textnormal{min}\geq3\widecheck{h}>3\widehat{h}/2$, hence
    \[
    L'(\ell h_\textnormal{min})\geq\left(1-\left(\frac23\right)^{\alpha-\beta}\right)\frac{\beta B}{(\ell h_\textnormal{min})^{\beta+1}}.
    \]
    Using $h_\textnormal{min}\leq\widehat{h}$ and comparing the series with an integral yields
    \begin{equation}
        L^\flat(h_\textnormal{min})\geq\left(1-\left(\frac23\right)^{\alpha-\beta}\right)\frac{\beta B}{\widehat{h}^{\beta+1}}\sum_{\ell=N+1}^{\infty}\ell^{-\beta}\geq\frac{c_0}{N^{\beta-1}},
    \end{equation}
    where $c_0>0$ is independent of $N$. The explicit formula for $L^\flat$ shows that $L^\flat(x)>0$ precisely when $x>\overline{h}$, hence $h_\textnormal{min}>\overline{h}$. Finally, with $M=\max_{[\overline{h},\widehat{h}]}(L^\flat)'>0$, the mean value theorem gives
    \[
    \frac{c_0}{N^{\beta-1}}\leq L^\flat(h_\textnormal{min})-L^\flat(\overline{h})\leq M(h_\textnormal{min}-\overline{h}),
    \]
    which proves the lower bound.
    }
\end{proof}

The first inequality of this lemma directly gives the first inequality of~\eqref{eq:asymptotics} in Theorem~\ref{thrm:spear}-$(iii)$.
The interest of the second inequality is that it provides an upper bound on the degeneracy of the distances between particles at the extremities of the structure. More precisely, in the view of Proposition~\ref{prop:no convergence}, this lemma implies
\blue{
\begin{equation}
    \limsup_{N\to+\infty}h^\ast_{N,k_0}\leq\widetilde{h}.
\end{equation}
}
The fail of convergence toward $\overline{h}$ as stated by Proposition~\ref{prop:no convergence} is then controlled by the distance between $\overline{h}$ and $\widetilde{h}$. For illustration (these parameters do not satisfy $\alpha>\alpha_\ast(\beta)$), if $A=B=1$ and $\alpha=12$, $\beta=3$, we have $\widetilde{h}\approx1.1563$ and $\overline{h}\approx1.1429$. In the case $\alpha=12$ and $\beta=6$, we have $\widetilde{h}\approx1.1209$ and $\overline{h}\approx1.1193$.

\subsubsection{Proof of Theorem~\ref{thrm:spear}-(iii) completed}
We now complete the proof of Theorem~\ref{thrm:spear}-(iii) by establishing the last inequality in~\eqref{eq:asymptotics}. 
This result is a direct consequence of the following lemma:

\blue{
\begin{lemma} \label{lem:convergence_centre}
    Let $N \in \mathbb N$ and let $H^\ast\in\RR^{N-1}$ be the minimizer of the function $J$ introduced at~\eqref{def:J L d} with $\alpha>\beta>1$. Assume moreover that $\alpha >\alpha_\ast(\beta)$. Then, there exists a constant $C=C(A,B,\alpha,\beta)$ independent of $N$ such that: 
    \begin{equation}\label{eq:holy graal}
       \forall\,k=1,\dots,N-1,\qquad | \overline{h}- h^*_k | \leq C \bigg( \frac{1}{k^{\beta -1}} + \frac{1}{(N-k)^{\beta -1}} \bigg)
    \end{equation}
    where $\overline{h}$ has been defined at~\eqref{def:the H gang}.
\end{lemma}
}

\begin{proof}
 \blue{
    Let $H^*\in\RR^{N-1}$ be the minimizer of $J$ and let $\overline{H}:= (\overline{h}, \dots, \overline{h})$. For $k\in\{1,\dots,N-1\}$, we compute:
    \begin{equation}
        \partial_{h_k}J(\overline{H})=\sum_{i=1}^k \sum_{j=k+1}^N L' ( (j-i) \;\!\overline{h}) = \sum_{i=1-k}^0 \sum_{j=1}^{N-k} L' ( (j-i) \;\!\overline{h}) 
    \end{equation}
    }
    \textit{Step 1.} Using the fact that $\overline{h}$ satisfies \eqref{eq:h_c 1}, we can continue the computation as follows:
    \begin{equation}
        \begin{split}
            \partial_{h_k}J(\overline{H})&=0-\sum_{i=-\infty}^{-k} \sum_{j=1}^{\infty} L' ( (j-i) \;\!\overline{h}) - \sum_{i=1-k}^0 \sum_{j=N-k+1}^{\infty} L' ( (j-i)\;\!\overline{h})\\
            &= - \sum_{i=-\infty}^{-k} \sum_{\ell=1-i}^{\infty} L' ( \ell \;\!\overline{h}) - \sum_{i=1-k}^0 \sum_{\ell=N-k+1-i}^{\infty} L' ( \ell\;\!\overline{h})
        \end{split}
    \end{equation}
    We can now swap the two sums appearing in the double sums above and get
    \begin{equation}
        \begin{split}
            \partial_{h_k}J(\overline{H}) &= - \sum_{\ell=1+k}^{\infty} \sum_{i=1-\ell}^{-k} L' ( \ell \;\!\overline{h}) - \sum_{\ell=N-k+1}^{\infty} \sum_{i=\max(1-k, N-k+1-\ell)}^0 L' ( \ell \;\!\overline{h}) \\
            &= - \sum_{\ell=1+k}^{\infty} (\ell - k) L' ( \ell \;\!\overline{h} ) - \sum_{\ell=N-k+1}^{\infty} (k - \max(0, N-\ell)) L' ( \ell \;\!\overline{h} ) \\
            &= - \sum_{\ell=1+k}^{\infty} (\ell - k) L' ( \ell \;\!\overline{h}) - \sum_{\ell=N-k+1}^{N} (\ell + k - N) L' ( \ell \;\!\overline{h}) - \sum_{\ell=N+1}^{\infty} k L' ( \ell \;\!\overline{h}).
        \end{split}
    \end{equation}
    
    Since $\ell \geq 2$ in each of those terms, we know that all the $L' ( \ell \overline{h})$ are positive, which leads to $\partial_{h_k} J (\overline{H}) \leq 0$. On the other hand, we also know that $L' (s) \leq \frac{\beta B}{s^{\beta + 1}}$, which implies
    \begin{equation}
        |\partial_{h_k} J (\overline{H})| \leq \sum_{\ell=1+k}^{\infty} (\ell - k) \frac{\beta B}{(\ell \overline{h})^{\beta+1}} + \sum_{\ell=N-k+1}^{N} (\ell + k - N) \frac{\beta B}{(\ell \overline{h})^{\beta+1}} + \sum_{\ell=N+1}^{\infty} k \frac{\beta B}{(\ell \overline{h})^{\beta+1}}
    \end{equation}
    Thus, \color{blue} using $k \leq \ell + k - N$ for $\ell \geq N + 1$,
    \begin{equation}
        |\partial_{h_k} J (\overline{H})|\leq\frac{\beta B}{\overline{h}^{\beta+1}}\bigl(\xi(k)+\xi(N-k)\bigr),
    \end{equation}
    \color{black}
        where $\xi (n) = \sum_{\ell = n+1}^\infty (\ell - n) \ell^{-(\beta + 1)}$.     
    Moreover, since we have $(\ell-n)\leq\ell$, then $\xi (n) \leq \frac{C_\beta}{n^{\beta - 1}}$ for some $C_\beta$ depending only on $\beta > 1$. We are led to:
    \begin{equation} \label{eq:est d_hk J}
        |\partial_{h_k} J (\overline{H})| \leq C_\beta \frac{\beta B}{\overline{h}^{\beta+1}} \Bigl( k^{1 - \beta} + (N-k)^{1 - \beta} \Bigr).
    \end{equation}   
    
  \noindent  \textit{Step 2.} On the other hand, with $H^*$ the minimizer of $J$, we have
    \begin{equation} \label{eq:exp_nablaJ_hessian}
        \nabla J (\overline{H}) = \nabla J (\overline{H}) - \nabla J (H^*) = \int_0^1 \nabla^2 J \Big(\overline{H} + t \big(H^* - \overline{H}\big)\Big) dt \, (\overline{H} - H^*),
    \end{equation}
    where $\nabla^2$ refers to the hessian matrix.    
\blue{
   For all $t \in (0, 1)$, there holds $\overline{H} + t (H^* - \overline{H}) \in [\widecheck{h}, \widehat{h}]^{N-1}$. Since $\alpha>\alpha_\ast(\beta)$, Lemma~\ref{lem:Hess_prop} applies along this segment. Moreover, the definition of $\alpha_G(\beta)$ ensures that Proposition~\ref{prop:matrix_inverse} applies with $\gamma:=\beta$, $d=\Lambda_d$ and $c=\Lambda_{nd}$. In particular, the matrix in \eqref{eq:exp_nablaJ_hessian} is invertible and we write
    \begin{equation}
        \overline{H} - H^* = \biggl( \int_0^1 \nabla^2 J \Big(\overline{H} + t \big(H^* - \overline{H}\big)\Big) dt \biggr)^{-1} \nabla J (\overline{H}).
    \end{equation}
    }
    In other words, looking coordinate by coordinate, 
    \begin{equation}
        \overline{h} - h^*_k = \sum_{\ell = 1}^{N-1} \biggl( \biggl( \int_0^1 \nabla^2 J \Big(\overline{H} + t \big(H^* - \overline{H}\big)\Big) dt \biggr)^{-1} \biggr)_{k \ell} \partial_{h_\ell} J (\overline{H}),
    \end{equation}
    which leads to
    \blue{
    \begin{equation}\begin{split}
        | \overline{h} - h^*_k | &\leq \sum_{\ell = 1}^{N-1} \Biggl| \biggl( \biggl( \int_0^1 \nabla^2 J \Big(\overline{H} + t \big(H^* - \overline{H}\big)\Big) dt \biggr)^{-1} \biggr)_{k \ell} \Biggr|\; \Bigl| \partial_{h_\ell} J (\overline{H}) \Bigr| \\
            &\leq \kappa \frac{1}{\Lambda_d} \Bigl| \partial_{h_k} J (\overline{H}) \Bigr| + \sum_{\ell \neq k} \kappa \frac{\Lambda_{nd}}{\Lambda_d^2 |\ell - k|^\beta} \Bigl| \partial_{h_\ell} J (\overline{H}) \Bigr|,
    \end{split}\end{equation}
    where for the last inequality we used Proposition~\ref{prop:matrix_inverse}.
    Thanks to \eqref{eq:est d_hk J}, this yields
    \begin{equation}\label{ouranos}
        | \overline{h} - h^*_k | \leq \kappa C_\beta \frac{1}{\Lambda_d} \frac{\beta B}{\overline{h}^{\beta+1}} \Bigl( k^{1 - \beta} + (N-k)^{1 - \beta} \Bigr) + \kappa C_\beta \frac{\Lambda_{nd}}{\Lambda_d^2}\frac{\beta B}{\overline{h}^{\beta+1}} \sum_{\ell \neq k} \frac{\ell^{1 - \beta} + (N-\ell)^{1 - \beta}}{|\ell - k|^\beta}.
\end{equation}
We now estimate the second term in the right-hand side. For every $\beta>1$, there exists $C_\beta>0$ such that, for all $1\leq k\leq N-1$,
\begin{equation}\label{eq:weighted_convolution}
    \sum_{\substack{\ell=1\\\ell\neq k}}^{N-1}\frac{\ell^{1-\beta}}{|\ell-k|^\beta}\leq C_\beta k^{1-\beta}.
\end{equation}
For $\ell<2k$, we have $\ell^{1-\beta}\leq2k\ell^{-\beta}$. Lemma~\ref{lem:sum shifted inverse}, applied with $\gamma=\beta$, $i=1$ and $j=k+1$, yields
\begin{equation}
    \sum_{\substack{1\leq\ell<2k\\\ell\neq k}}\frac{\ell^{1-\beta}}{|\ell-k|^\beta}\leq2\delta k^{1-\beta}.
\end{equation}
For $\ell\geq2k$, we have $|\ell-k|\geq\ell/2$, and therefore
\begin{equation}
    \sum_{\ell\geq2k}\frac{\ell^{1-\beta}}{|\ell-k|^\beta}\leq C_\beta\sum_{\ell\geq2k}\ell^{1-2\beta}\leq C_\beta k^{2-2\beta}\leq C_\beta k^{1-\beta}.
\end{equation}
This proves~\eqref{eq:weighted_convolution}. By the change of index $\ell\mapsto N-\ell$, we also obtain
\begin{equation}
    \sum_{\ell\neq k}\frac{(N-\ell)^{1-\beta}}{|\ell-k|^\beta}\leq C_\beta(N-k)^{1-\beta},
\end{equation}
Plugging these estimates into~\eqref{ouranos} gives~\eqref{eq:holy graal}.
}
\end{proof}
\blue{
\begin{remark} \label{rk:alpha*}
The definition~\eqref{eq:def_alpha_master} gathers all restrictions on $\alpha$ used in Theorem~\ref{thrm:spear}-$(iii)$. The threshold $\alpha_H$ guarantees the monotonicity properties of $L''$ needed in Lemma~\ref{lem:Hess_prop}, $\alpha_Z$ implies both inequalities in~\eqref{eq:alpha_Z_consequences}, and $\alpha_G$ is the condition required to apply Proposition~\ref{prop:matrix_inverse} with $\gamma=\beta$. Moreover, Lemma~\ref{lem:a*>ad} gives $\alpha_G(\beta)>\alpha_\dag(\beta)$, so the same assumption also yields the strong convexity required for uniqueness.

For the main case of physical interest $\beta=3$, one finds numerically $\alpha_\dag(3)\approx4.9$, $\alpha_H(3)\approx3.67$, $\alpha_Z(3)\approx3.81$, and $\alpha_G(3)\approx33.21$. Hence $\alpha_\ast(3)=\alpha_G(3)\approx33.21$. For $\beta=6$, corresponding to the Van der Waals scaling, $\alpha_\dag(6)\approx6.3$, $\alpha_H(6)\approx6.04$, $\alpha_Z(6)\approx6.12$, and $\alpha_G(6)\approx160.51$, so again $\alpha_\ast(6)=\alpha_G(6)\approx160.51$.
\end{remark}
}
\section{Study of the ring structure }
We now study the ring structure defined at~\eqref{def:ring} and prove Theorem~\ref{thrm:ring}. 

\subsection{Existence and Uniqueness}

To prove the existence of a critical point that have a ring structure, we start by computing the gradient at this position $\mR_r\in\big(\RR^3\times\SS^2\big)^N$.

\begin{lemma}\label{lem:grad ring}
Let $r>0$. The gradient of the energy at the ring $\mR_r$ without the constraint $|m_i|=1$ is given by:
\begin{equation}\begin{split}
    &\nabla_{x_i} U(\mR_r)=\nabla_{x_i}U^d(\mR_r)+\nabla_{x_i}U^s(\mR_r) \\
    &\nabla_{m_i} U(\mR_r)= \nabla_{m_i}U^d(\mR_r)
    \end{split}
\end{equation}
where 
\begin{equation}\label{eq:ring energy gradient position}
    \nabla_{x_i} U^d(\mR_r)=\frac{\beta}{2^{\beta+1}\,r^{\beta+1}}\sum_{j=1}^{N-1}\frac{1}{\big|\sin\big(\frac{j\pi}{N}\big)\big|^\beta}
    \bigg[(B - B_0)\cos^2\left(\frac{j\pi}{N}\right)+B_0\bigg]\frac{x_i}{|x_i|},
\end{equation}
\begin{equation}
\nabla_{x_i}U^s(\mR_r)= -\frac{\alpha A}{2^{\alpha+1} r^{\alpha+1}}\sum_{j=1}^{N-1}\frac{1}{\Big|\sin\left(\frac{j\pi}{N}\right)\Big|^{\alpha}}\frac{x_i}{|x_i|}
\end{equation}
and 
\begin{equation}\label{eq:ring energy gradient spin}
    \nabla_{m_i}U^d(\mR_r)=\frac{1}{2^\beta r^\beta}\sum_{j=1}^{N-1}\Bigg(\frac{B-B_0}{\big|\sin\Big(\frac{j\pi}{N}\Big)\big|^{\beta-2}}-\frac{B}{\big|\sin\Big(\frac{j\pi}{N}\Big)\big|^\beta}\Bigg)\frac{m_i}{|m_i|}.
\end{equation}
\end{lemma}

Gradient computations are based on explicit derivation of the dipolar interaction potential~\eqref{eq:interaction potential} and manipulations involving the $N$-fold symmetry property. 
To ease the reading of the article, we postponed these computations to the appendix.
From this lemma, we are able to establish the existence of the ring;
\begin{lemma}\label{lem:radius ring}
    There exists a unique radius $r^\star$ such that the ring $\mR_{r^\star}$ is a critical point of the potential energy $U^d+U^s$ under the magnetization constraint $|m_i|=1$. Moreover, there holds 
    \begin{equation} \label{def:rN}
    r^\ast_N=(\widetilde{A}_N/\widetilde{B}_N)^{1/(\alpha-\beta)},
    \end{equation}
    where 
    \begin{equation} \label{eq:def_A_tilde}
            \widetilde{A}_N:=\frac{\alpha A}{2^{\alpha+1}}\sum_{j=1}^{N-1}\frac{1}{\Big|\sin\left(\frac{j\pi}{N}\right)\Big|^{\alpha}}\end{equation}
    and
   \begin{equation}  \label{eq:def_B_tilde}
    \widetilde{B}_N:=\frac{\beta}{2^{\beta+1}}\sum_{j=1}^{N-1}\frac{1}{\big|\sin\big(\frac{j\pi}{N}\big)\big|^\beta}
    \bigg[(B - B_0)\cos^2\left(\frac{j\pi}{N}\right)+B_0\bigg].
    \end{equation}
\end{lemma}
\begin{proof}
    To start with, we observe that in the absence of the constraint $|m_i|=1$ the gradient of $U^d+U^s$ with respect to $m_i$  is a vector colinear with $m_i$. Therefore, in the presence of this constraint, it has a zero contribution (recall that the gradient of $m_i\mapsto|m_i|$ is colinear to $m_i$) by virtue of the Lagrange multiplier theorem.
    In reality, the coefficient appearing in front of $m_i/|m_i|$ in~\eqref{eq:ring energy gradient spin} is the Lagrange multiplier associated to the constraint $|m_i|=1$.

    We now focus on the gradient with respect to $x_i$. Given the formulas~\eqref{eq:ring energy gradient position}, we can reduce the analysis to the existence and uniqueness of a critical radius $r=r^\ast_N$, which is a zero point for the following function defined on $(0,+\infty)$:
    \begin{equation}\label{de Cadix}
        r\;\longmapsto\;\frac{\widetilde{A}_N}{r^{\alpha+1}}-\frac{\widetilde{B}_N}{r^{\beta+1}},
    \end{equation}
    where $\widetilde{A}_N$ and $\widetilde{B}_N$ are defined in \eqref{eq:def_A_tilde} and \eqref{eq:def_B_tilde}.
    It is direct to check that the function~\eqref{de Cadix} admits a unique zero provided that $\widetilde{A}_N$ and $\widetilde{B}_N$ are positive.
    It is clear that $\widetilde{A}_N>0$. Regarding $\widetilde{B}_N$, we can conclude with the following computation:
    \begin{equation}\begin{split}
        (B - B_0)\cos^2\left(\frac{j\pi}{N}\right)+B_0 &= (B - B_0)\cos^2\left(\frac{j\pi}{N}\right)+B_0 \left( \sin^2\left(\frac{j\pi}{N}\right) + \cos^2\left(\frac{j\pi}{N}\right) \right) \\
            &= B \cos^2\left(\frac{j\pi}{N}\right)+B_0\sin^2\left(\frac{j\pi}{N}\right) > 0.
    \end{split}\end{equation}

We note that the zero of the function~\eqref{de Cadix} is explicitly given by $r^\ast_N=(\widetilde{A}_N/\widetilde{B}_N)^{1/(\alpha-\beta)}$.
\end{proof}



\subsection{Asymptotic analysis}

Now, we turn our attention to the asymptotics of the quantities mentioned in Theorem \ref{thrm:ring}-$(ii)$. In view of the expression of $r^\ast_N$, we only have to get an expansion for $\widetilde{A}_N$ and $\widetilde{B}_N$. 

\begin{proof}[Proof of Theorem \ref{thrm:ring}-$(ii)$]

First, recall two elementary estimates: 
    there exists $C > 0$ such that, for all $x \in (0, \frac{\pi}{2}]$, there holds
    \begin{align}
        \biggl| \frac{1}{| \sin (x) |^\alpha} - \frac{1}{x^\alpha} \biggr| \leq C x^{2 - \alpha}, \label{est:sin1}\\
        \biggl| \frac{1}{\big|\sin(x)\big|^\beta}
    \bigg[(B - B_0)\cos^2(x)+B_0\bigg]  - \frac{B}{x^\beta} \biggr| \leq C x^{2 - \beta}. \label{est:sin2}
    \end{align}
They are obtained directly from standard function analysis and Taylor expansions.\medskip

\noindent \textit{Step 1}: Asymptotic of $\widetilde{A}_N$.
    
 We start by splitting the sum in $\widetilde{A}_N$ into two separate sums and utilizing the $\pi$-periodicity of $\sin^2$ to obtain:
    \begin{equation}
        \widetilde{A}_N = \frac{\alpha A}{2^{\alpha+1}} \Biggl( \sum_{j=1}^{\lfloor \frac{N}{2} \rfloor}\frac{1}{\Big|\sin\left(\frac{j\pi}{N}\right)\Big|^{\alpha}} + \sum_{j=1}^{N-\lfloor \frac{N}{2} \rfloor-1}\frac{1}{\Big|\sin\left(\frac{j\pi}{N}\right)\Big|^{\alpha}} \Biggr).
    \end{equation}
    Using \eqref{est:sin1}, we get for the first term:
    \begin{equation}
        \Bigl| \Xi_N^\alpha \Bigr| \leq C \sum_{j=1}^{\lfloor \frac{N}{2} \rfloor}\frac{N^{\alpha-2}}{(j\pi)^{\alpha-2}} 
            \leq C \frac{N^{\alpha-2}}{\pi^{\alpha-2}} \sum_{j=1}^{\lfloor \frac{N}{2} \rfloor}\frac{1}{j^{\alpha-2}},
    \end{equation}
    where we have set
    \begin{equation}
        \Xi_N^\alpha := \sum_{j=1}^{\lfloor \frac{N}{2} \rfloor}\frac{1}{\Big|\sin\left(\frac{j\pi}{N}\right)\Big|^{\alpha}} - \sum_{j=1}^{\lfloor \frac{N}{2} \rfloor}\frac{N^\alpha}{(j\pi)^{\alpha}}
    \end{equation}
    If $\alpha > 3$, then $\sum_{j=1}^{\lfloor \frac{N}{2} \rfloor} j^{2 - \alpha}$ is bounded and we get 
    \begin{equation*} \label{eq:asympt_exp_A_N}
        \Bigl| \Xi_N^\alpha \Bigr| \leq C N^{\alpha - 2}.
    \end{equation*}
    If $\alpha = 3$, then $\sum_{j=1}^{\lfloor \frac{N}{2} \rfloor} j^{2 - \alpha} \leq C \log{N}$ and 
    \begin{equation*} \label{eq:asympt_exp_A_N_2}
        \Bigl| \Xi_N^\alpha \Bigr| \leq C N \log{N}.
    \end{equation*}
    If $1 < \alpha < 3$, then $\sum_{j=1}^{\lfloor \frac{N}{2} \rfloor} j^{2 - \alpha} \leq C N^{3 - \alpha}$ and so
    \begin{equation*} \label{eq:asympt_exp_A_N_3}
        \Bigl| \Xi_N^\alpha \Bigr| \leq C N.
    \end{equation*}
    On the other hand, there also holds
    \begin{equation} \label{eq:zeta_like_exp}
        \sum_{j=1}^{\lfloor \frac{N}{2} \rfloor}\frac{N^\alpha}{(j\pi)^{\alpha}} = \frac{N^\alpha}{\pi^{\alpha}} \sum_{j=1}^{\lfloor \frac{N}{2} \rfloor}\frac{1}{j^{\alpha}} 
        = \frac{N^\alpha}{\pi^{\alpha}} \zeta(\alpha) + O(N),
        \quad \text{as } N \to \infty.
    \end{equation}
    Concerning the second term, we observe that it can be studied using a similar technique, which eventually leads to the same asymptotic quantity.
    This leads to
    \begin{equation}\label{tac}
        \widetilde{A}_N \sim 2 \frac{\alpha A}{2^{\alpha+1}} \frac{N^\alpha}{\pi^{\alpha}}\zeta(\alpha) = \frac{\alpha A N^\alpha}{(2\pi)^{\alpha}}\zeta(\alpha), \quad \text{as } N \to \infty,
    \end{equation}
    where the notation $f(N)\sim g(N)$ for two functions $f,g:\NN\to\RR$ means that $f(N)\big/g(N)\longrightarrow1$ as $N\to+\infty$.
    \medskip

 \noindent \textit{Step 2}: asymptotic of $\widetilde{B}_N$

   Similarly to the study of $\widetilde{A}_N$, we split into two parts the summation defining $\widetilde{B}_N$:
    \begin{equation}\label{Dumas}
        \widetilde{B}_N = \frac{\beta}{2^{\beta+1}} \Biggl( \sum_{j=1}^{\lfloor \frac{N}{2} \rfloor} + \sum_{j=1}^{N-\lfloor \frac{N}{2} \rfloor-1} \Biggr)\biggl[\frac{1}{\big|\sin\big(\frac{j\pi}{N}\big)\big|^\beta}
    \Big((B - B_0)\cos^2\left(\frac{j\pi}{N}\right)+B_0\Big)\biggr].
    \end{equation}
We first observe that
    \begin{equation}\begin{split}
        \bigl|\Xi_N^\beta\bigr|\leq C \sum_{j=1}^{\lfloor \frac{N}{2} \rfloor} \frac{N^{\beta - 2}}{(j \pi)^{\beta - 2}}                \leq C \frac{N^{\beta - 2}}{\pi^{\beta - 2}} \sum_{j=1}^{\lfloor \frac{N}{2} \rfloor} j^{2 - \beta},
    \end{split}\end{equation}
        where we have set
    \begin{equation}
        \label{Richelieu}
        \Xi_N^\beta:=\sum_{j=1}^{\lfloor \frac{N}{2} \rfloor} \frac{1}{\big|\sin\big(\frac{j\pi}{N}\big)\big|^\beta} \bigg[(B - B_0)\cos^2\left(\frac{j\pi}{N}\right)+B_0\bigg] - B \sum_{j=1}^{\lfloor \frac{N}{2} \rfloor} \frac{N^\beta}{(j \pi)^\beta}
    \end{equation}
    If $\beta > 3$, then $\sum_{j=1}^{\lfloor \frac{N}{2} \rfloor} j^{2 - \beta}$ is bounded and then
    \begin{equation*}\label{Porthos}
        \bigl|\Xi_N^\beta\bigr| \leq C N^{\beta - 2}.
    \end{equation*}
    If $\beta = 3$, then $\sum_{j=1}^{\lfloor \frac{N}{2} \rfloor} j^{2 - \beta} \leq C \log{N}$, which yields
    \begin{equation*}\label{Athos}
                \bigl|\Xi_N^\beta\bigr| \leq C N \log{N}.
    \end{equation*}
    If $\beta < 3$, then $\sum_{j=1}^{\lfloor \frac{N}{2} \rfloor} j^{2 - \beta} \leq C N^{3 - \beta}$ and we get
    \begin{equation*}\label{Aramis}
               \bigl|\Xi_N^\beta\bigr| \leq C N.
    \end{equation*}
    On the other hand, arguing similarly as for \eqref{eq:zeta_like_exp}, we have
        \begin{equation*} \label{eq:zeta_like_exp_2}
        \sum_{j=1}^{\lfloor \frac{N}{2} \rfloor}\frac{N^\beta}{(j\pi)^{\beta}} = \frac{N^\beta}{\pi^{\beta}} \zeta(\beta) + O(N)
        \quad \text{as } N \to \infty.
    \end{equation*}
Since the second term in~\eqref{Dumas} can be estimated with the same technique as for $\widetilde{A}_N$, we get as $N \to \infty$,
    \begin{equation}\label{tic}
        \widetilde{B}_N \sim 2 \frac{\beta B}{2^{\beta+1}} \frac{N^\beta}{\pi^{\beta}} \zeta(\beta) = \frac{\beta B N^\beta}{(2 \pi)^{\beta}} \zeta(\beta).
    \end{equation}
    \medskip

  \noindent  \textit{Step 3}: The previous asymptotics on $\widetilde{A}_N$ and on $\widetilde{B}_N$ at~\eqref{tac} and~\eqref{tic} give the announced asymptotic 
    \begin{equation}\label{D Artagnan}
        r_N^\ast\,\sim\,\frac{N}{2\pi}\overline{h} \qquad \text{as } N \to \infty
    \end{equation}
    directly from the formula $r_N^\ast = (\widetilde{A}_N/\widetilde{B}_N)^{1/(\alpha-\beta)}$. Concerning the distance between two nearest particles, it is given (for instance) by the distance between $x_1$ and $x_N$:
    \begin{align}\label{vivement Dimanche}
        | x_1 - x_N | &= r^\ast_N \Biggl| \begin{pmatrix}
                \cos\left(\frac{2\pi}{N}\right) - 1\\
                \sin\left(\frac{2\pi}{N}\right)\\
                0
            \end{pmatrix} \Biggr|= r^\ast_N \sqrt{\Bigl(\cos\left(\frac{2\pi}{N}\right) - 1\Bigr)^2 + \sin^2\left(\frac{2\pi}{N}\right)} = r^\ast_N \frac{2\pi}{N} (1 + O(N^{-2})),
    \end{align}
    in the asymptotics $N \to \infty$. Thus, using~\eqref{D Artagnan}, we get the convergence of the distances between two neighboring nanoparticle towards $\overline{h}$.
    Concerning the speed of convergence towards $\overline{h}$, as stated by~\eqref{Milady de Winter}, 
    they are obtained from \eqref{eq:asympt_exp_A_N}, \eqref{eq:asympt_exp_A_N_2}, \eqref{eq:asympt_exp_A_N_3}, \eqref{eq:zeta_like_exp}, \eqref{Porthos},\eqref{Athos}, \eqref{Aramis} and \eqref{eq:zeta_like_exp_2}, so that
    \begin{equation*}
        r_N^* = \frac{\overline{h}}{2 \pi} \Bigl( \frac{N^\alpha + O(N) + O(\Xi_N^\alpha)}{N^\beta + O(N) + O(\Xi_N^\beta)} \Bigr)^\frac{1}{\alpha - \beta},
    \end{equation*}
    and straightforward elementary asymptotic expansions, using $\alpha > \beta$.
\end{proof}

\blue{
\section*{Conclusion}

We have studied two stationary configurations of the magnetic-colloid model. For the spear, we prove existence of a minimizing configuration within the spear ansatz, uniqueness up to translations under an explicit condition on $\alpha$, and quantitative estimates on the neighboring-particle distances. These estimates give the bulk limit as $N\to+\infty$ and show that convergence is not uniform up to the endpoints. For the ring, we obtain a unique critical radius and its large-$N$ asymptotics.

The paper does not address dynamical stability or convergence of the full dissipative dynamics. For the ring, we prove criticality in the full constrained configuration space, but no minimization result.
}

\subsection*{Acknowledgements}

The author acknowledges partial grant support by the ANR project MOSICOF ANR-21-CE40-0004. R. Côte has benefited from support provided by the University of Strasbourg Institute for Advanced Study (USIAS) for a Fellowship within the French national program ``Investment for the future'' (IdEx-Unistra).
\appendix
\section{Appendix}

\subsection{Quantitative Gershgorin Circles Theorem}


We start with two technical lemmas.
\blue{
\begin{lemma}\label{lem:sum shifted inverse}
Let $N\in\m N^*$, let $i,j\in\{1,\dots,N\}$ with $i \ne j$, and let $\gamma>1$. Then for $\delta := 2(1+2^\gamma) \zeta(\gamma)$ there holds
\begin{equation} \sum_{k \in \llbracket 1, N \rrbracket \setminus \{ i,j \}} \frac{1}{|i-k|^\gamma} \frac{1}{|k-j|^\gamma} 
\le \frac{\delta}{|i-j|^\gamma}. \end{equation}
\end{lemma}
}

\begin{proof}
Without restriction of generality, it suffices to consider the case $i < j$. First, we have
\begin{equation}\begin{split}
\sum_{k=1}^{i-1}  \frac{1}{(i-k)^\gamma} \frac{1}{(j-k)^\gamma}  = \sum_{h=1}^{i-1} \frac{1}{h^\gamma}  \frac{1}{(j-i+h)^\gamma}  \le \frac{1}{(j-i)^\gamma} \sum_{h=1}^{+\infty} \frac{1}{h^\gamma} \le \frac{\zeta(\gamma)}{(j-i)^\gamma}.
\end{split}\end{equation}
Similarly,
\blue{
\begin{equation}\begin{split}
\sum_{k=j+1}^{N}  \frac{1}{(k-i)^\gamma} \frac{1}{(k-j)^\gamma}  = \sum_{h=1}^{N-j} \frac{1}{(j-i+h)^\gamma}  \frac{1}{h^\gamma} \le \frac{1}{(j-i)^\gamma} \sum_{h=1}^{+\infty} \frac{1}{h^\gamma} \le \frac{\zeta(\gamma)}{(j-i)^\gamma}.
\end{split}\end{equation}
}
Secondly, splitting the middle sum around $(j+i)/2$, there hold
\begin{equation}\begin{split}
\sum_{k=i+1}^{j-1} \frac{1}{(k-i)^\gamma} \frac{1}{(j-k)^\gamma} & \le 2 \sum_{k=i+1}^{\lfloor \frac{j+i}{2} \rfloor}  \frac{1}{(k-i)^\gamma} \frac{1}{(j-k)^\gamma} \\
& \le 2 \sum_{k=i+1}^{ \lfloor \frac{j+i}{2} \rfloor}  \frac{1}{(k-i)^\gamma} \frac{2^\gamma}{(j-i)^\gamma}  \le \frac{2^{\gamma+1} \zeta(\gamma)}{(j-i)^\gamma}
\end{split}\end{equation}
Summing up the three bounds yields the result
\end{proof}

Let us recall the constant $r^+$ defined in \eqref{def:r+}
\[
r_+ = \frac{c}{d} \frac{\delta + \sqrt{\delta^2 + 8 \zeta(2\gamma)}}{2}.
\]
\begin{lemma}\label{lem:the previous lemma} With the notations of Proposition~\ref{prop:matrix_inverse}, define from $A$ two matrices $D$ and $B$ as follows:
\begin{equation}
D:= \mathrm{diag}(A_{ii})_{i=1}^N,\qquad\text{and}\qquad B:= I_N - D^{-1}A.\end{equation}
\blue{
Then for all $k \in \m N^*$:
\begin{equation}
\forall i, j = 1, \dots, N, \qquad |(B^k)_{ii}| \le 2 \zeta (2\gamma) \bigg(\frac{c}{d}\bigg)^2 r_+^{k-2},\qquad\text{and}\qquad |(B^k)_{ij}| \le \frac{c\, r_+^{k-1}}{d |i-j|^\gamma},\quad\text{if}\quad i \ne j.\end{equation}
}
\end{lemma}

\begin{proof}
Notice  that $B_{ii}=0$ and
\begin{equation} |B_{ij}| \le  \frac{c}{d |i-j|^\gamma} \quad \text{for } i \ne j. \end{equation}

\textit{Step 1.} We show by induction on $k \ge 1$ that
\begin{equation} |(B^k)_{ij}| \le \frac{c_k}{|i-j|^\gamma} \quad \text{for } i \ne j, \end{equation}
and
\begin{equation}
    |(B^k)_{ii}| \leq c_k',
\end{equation}
where $c_k$ is defined by $c_1 = \frac{c}{d}$ and $c_{k+1} =  \frac{c\delta}{d} c_k + \frac{c}{d} c_k'$ and $c_k'$ is defined by $c_1' = 0$ and $c_{k+1}' = 2 \frac{c}{d} \zeta (2 \gamma) c_k$.
For $k=1$, the result is clear.
\blue{
Assume it holds at rank $k \ge 1$. Then using $B_{ii} =0$, we get
\begin{equation}
    (B^{k+1})_{ij} = (B \cdot B^{k})_{ij} = \sum_{\ell =1}^N B_{i \ell} (B^k)_{\ell j} = \sum_{\ell \neq i, j} B_{i \ell} (B^k)_{\ell j} + B_{i j} (B^k)_{j j}.
\end{equation}
If $i=j$, then the last term vanishes and we get
}
\begin{equation}
    |(B^{k+1})_{ii}| \leq \sum_{\ell \ne i}  \frac{c}{d |i-\ell|^\gamma} \frac{c_k}{|\ell-i|^\gamma} \leq c_{k+1}'.
\end{equation}
If $i\neq j$, the previous lemma yields
\begin{equation}\begin{split}
|(B^{k+1})_{ij}| & \le \sum_{\ell \ne i,j}  \frac{c}{d |i-\ell|^\gamma} \frac{c_k}{|\ell-j|^\gamma} + \frac{c}{d |i-j|^\gamma} \, c_k' \\
& \le \left( \frac{c\delta}{d} c_k +  \frac{c}{d} c_k' \right)  \frac{1}{|i-j|^\gamma} = \frac{c_{k+1}}{|i-j|^\gamma}.
\end{split}\end{equation}
This completes the induction.

\textit{Step 2.} Now, we get 
\begin{equation}
c_{k+1} =  \frac{c\delta}{d} c_k + 2 \bigg(\frac{c}{d}\bigg)^2 \, \zeta (2 \gamma) \, c_{k-1},
\end{equation}
with the convention $c_0 = 0$. We also have $c_1 = \frac{c}{d}$. The definition of $r_+ > 0$ is such that
\begin{equation}
r_+^2 = \frac{c\delta}{d} r_+ + 2 \bigg( \frac{c}{d} \bigg)^2 \zeta (2 \gamma).
\end{equation}
We can then prove by induction that
\begin{equation}
    c_k \leq \frac{c}{d} r_+^{k-1},
\end{equation}
which completes the proof.
\end{proof}


\begin{proof}[Proof of Proposition~\ref{prop:matrix_inverse}]
\blue{
Let $D=\mathrm{diag}(A_{ii})$. Then $\|D^{-1}\|_\infty\leq d^{-1}$. Set $B=I_N-D^{-1}A$, so that $A=D(I_N-B)$, $B_{ii}=0$ and
\begin{equation}
    |B_{ij}|\leq\frac{c}{d|i-j|^\gamma},\qquad i\neq j.
\end{equation}
For the induced $\ell^\infty$ operator norm,
\begin{equation}
    \|M\|_\infty=\max_i\sum_{j=1}^N|M_{ij}|,
\end{equation}
and therefore
\begin{equation}
    b:=\|B\|_\infty\leq2\frac{c}{d}\zeta(\gamma)<1.
\end{equation}
}
where we used for the last inequality the assumption \eqref{def:r+} on $r_+$ and the fact that $\delta > 4 \zeta (\gamma)$.
Hence we can write
\begin{equation} A^{-1} = (I_N - B)^{-1} D^{-1} = \left( \, \sum_{k=0}^{+\infty} B^k \right) D^{-1}. \end{equation}
From Lemma~\ref{lem:the previous lemma} and the assumption \eqref{def:r+} on $r_+$, 
 we infer that for $i \ne j$,
\begin{equation}
\Bigg| \bigg(  \sum_{k=0}^{+\infty} B^k \bigg)_{ij} \Bigg| \le \frac{c}{d |i-j|^\gamma} \sum_{k=1}^{+\infty} r_+^{k-1} \le \frac{c}{d |i-j|^\gamma} \frac{1}{1-r_+},
\end{equation}
and so
\begin{equation} \bigg| (A^{-1})_{ij} \bigg| \le \frac{c}{d^2 |i-j|^\gamma} \frac{1}{1-r_+}.
\end{equation}
\blue{
On the other hand, if $i=j$,
\begin{equation}
\Bigg| \biggl( \sum_{k=0}^{+\infty} B^k \biggr)_{\! ii} \Bigg| = \Bigg| 1 + \biggl( \sum_{k=2}^{+\infty} B^k \biggr)_{\! ii} \Bigg| \le 1 + 2\zeta (2\gamma) \biggl(\frac{c}{d}\biggr)^{\! 2} \sum_{k=0}^{+\infty} r_+^k \le 1 + 2\zeta (2\gamma) \frac{1}{1-r_+} \biggl(\frac{c}{d}\biggr)^{\! 2},
\end{equation}
and therefore
\begin{equation*}
    \bigg| (A^{-1})_{ii} \bigg| \leq d^{-1} + 2\zeta (2\gamma) \frac{1}{1-r_+} \frac{c^2}{d^3} \leq \kappa d^{-1}.  \qedhere
\end{equation*}
}
\end{proof}

\subsection{Auxiliary lemmas used in the proofs of Proposition~\ref{prop:no convergence} and Theorem~\ref{thrm:spear}-(iii)}\label{append:zeta}

\begin{lemma}\label{lem:widecheck}
    Let $\alpha>\beta>1$, let $\widehat{h}$ and $\overline{h}$ be defined by~\eqref{def:the H gang} and $\widetilde{h}$ by~\eqref{def:h_b}. We have $\overline{h}<\widetilde{h}<\widehat{h}$.
\end{lemma}
\begin{proof}
    To start with, we recall the famous Von Mangoldt formula for the derivative of the zeta function:
    \begin{equation}\label{Von Mangoldt}
        \zeta'(x)\;=\;\zeta(x)\Phi(x),\qquad\text{with}\qquad\Phi(x):=-\sum_{n=1}^{\infty}\frac{\Lambda(n)}{n^x},
    \end{equation}
    where $\Lambda$ designates the so-called Von Mangoldt function:
    \begin{equation}
        \Lambda(n) = \begin{cases} \log p & \text{\quad if }n=p^k \text{ for some prime number } p \text{ and some integer } k \ge 1, \\ 0 & \text{\quad otherwise.} \end{cases}
    \end{equation}
  A direct computation yields
  \begin{equation}
        \Phi^{(i)}(x)\;=\;(-1)^{i+1}\sum_{n=1}^{\infty}\frac{\Lambda(n)}{n^x}\big(\log(n)\big)^i.
    \end{equation}
    We note in particular that we have $\Phi'(x)>0$ when $x\in (1,+\infty).$
    Therefore,
    \begin{equation}
        \frac{d}{dx}\frac{\zeta'(x)}{\zeta(x)}\;>\;0.
    \end{equation}
    We integrate this inequality and we deduce that for any $\lambda>0$:
    \begin{equation}
        \frac{\zeta'(x+\lambda)}{\zeta(x+\lambda)}>\frac{\zeta'(x)}{\zeta(x)}.
    \end{equation}
Thus,    
    \begin{equation}
        \frac{d}{dx}\log\bigg(\frac{\zeta(x+\lambda)}{\zeta(x)}\bigg)\;=\;\frac{\zeta'(x+\lambda)}{\zeta(x+\lambda)}-\frac{\zeta'(x)}{\zeta(x)}\;>\;0.
    \end{equation}
    From this monotony property we infer, since $\alpha>\beta$,
    \begin{equation}\frac{\zeta(\alpha)}{\zeta(\beta)}\;<\;\frac{\zeta(\alpha+1)}{\zeta(\beta+1)}. \end{equation}
    We conclude that
    \begin{equation}\overline{h}=\bigg(\frac{\alpha A \zeta(\alpha)}{\beta B\zeta(\beta)}\bigg)^\frac{1}{\alpha-\beta}\;<\;\bigg(\frac{\alpha A \zeta(\alpha+1)}{\beta B\zeta(\beta+1)}\bigg)^\frac{1}{\alpha-\beta}\;=\;\widetilde{h}.\end{equation}
    With the same arguments, since $\zeta(s)\to1$ as $s\to+\infty$, we also have
    \begin{equation*}
        \widetilde{h}\;<\;\bigg(\frac{\alpha A}{\beta B}\bigg)^\frac{1}{\alpha-\beta}\;=\;\widehat{h}.  \qedhere
    \end{equation*}
\end{proof}

\blue{
\begin{lemma} \label{lem:a*>ad}
For any $\beta>1$, there holds $\alpha_G(\beta) > \alpha_\dag(\beta)$, and consequently $\alpha_\ast(\beta)>\alpha_\dag(\beta)$.
\end{lemma}
}
\begin{proof}
Observe that the function $G_\beta$ (defined in \eqref{eq:def_alpha_ast}) is increasing \textcolor{blue}{on $(\beta, \infty)$}. Therefore, it suffices to prove that $G_\beta(\alpha_\dag(\beta)) <0$.

\blue{
Now, from the definition via $F_\beta$ (see \eqref{eq:def_alpha_dag}), there holds
\begin{equation*}\frac{\alpha_\dag - \beta}{(\beta+1) \zeta(\beta)^{\frac{\beta+2}{\alpha_\dag -\beta}}} = \zeta(\beta+1) + \zeta(\beta) -2.\end{equation*}
Hence
\begin{equation*}G_\beta(\alpha_\dag) = \frac{2}{\delta + \sqrt{\delta^2 + 8 \zeta(2\beta)}} \left( \beta( \zeta(\beta+1) + \zeta(\beta) -2) - \beta(\zeta(\beta+1)-1) \right) -1\end{equation*}
has the sign of
\begin{equation*}\beta(\zeta(\beta) -1) - \frac{\delta + \sqrt{\delta^2 + 8 \zeta(2\beta)}}{2} \leq \beta \zeta(\beta) - \delta.\end{equation*}
Finally recall that $\delta = 2(1+2^\beta) \zeta(\beta)$: the last expression is thus negative, and so $G_\beta(\alpha_\dag(\beta)) <0$.
}
\end{proof}
\subsection{Computation of the gradient for the ring}
In this appendix, we give the explicit computations that give Lemma~\ref{lem:grad ring}.
\begin{proof} 
    \textit{Step 1.}
To start with, we compute here the gradient with respect to $x_i$ of the dipolar interaction energy $U^d$.
Invoking the $N$-fold symmetry of the ring configuration $\mR_r$~\eqref{def:ring}, it is enough to compute $\nabla_{x_N} U^d(\mR_r)$.
In the specific case of the ring, the manipulated quantities write:
\begin{equation}\label{Paris}\begin{split}
&m_j=\begin{pmatrix}
    -\sin\left(\frac{2j\pi}{N}\right)\\
    \cos\left(\frac{2j\pi}{N}\right)\\
    0
\end{pmatrix}\\
    &r_{Nj}=x_N-x_j=r\begin{pmatrix}
    1-\cos\left(\frac{2j\pi}{N}\right)\\
    -\sin\left(\frac{2j\pi}{N}\right)\\
    0
\end{pmatrix}=r\begin{pmatrix}
    2\sin^2\left(\frac{j\pi}{N}\right)\\
    -\sin\left(\frac{2j\pi}{N}\right)\\
    0
\end{pmatrix},\\
    &|r_{Nj}|=r\sqrt{2-2\cos\Big(\frac{2j\pi}{N}\Big)\,}=2r\Big|\sin\Big(\frac{j\pi}{N}\Big)\Big|,\\
    &m_N\cdot r_{Nj}=-r\sin\Big(\frac{2j\pi}{N}\Big),\\
   & m_j\cdot r_{Nj}= -r\sin\Big(\frac{2j\pi}{N}\Big)
    ,\\
    &m_N\cdot m_j=\cos\Big(\frac{2j\pi}{N}\Big).
\end{split}
\end{equation}
The gradient of the potential dipolar interaction energy~\eqref{def:U d s} with respect to $x_N$ writes: 
\begin{multline}\label{Vers_ailles}
\nabla_{x_N}U^d=-\sum_{j=1}^{N-1}\frac{B+B_0}{|r_{Nj}|^{\beta+2}}\bigg[(m_N\cdot r_{Nj}) m_j+(m_j\cdot r_{Nj})m_N\\
+\frac{\beta B_0}{B+B_0}(m_N\cdot m_j)r_{Nj}-(\beta+2)\frac{(m_N\cdot r_{Nj})(m_j\cdot r_{Nj})}{|r_{Nj}|^2}r_{Nj}\bigg].
\end{multline}
If we now inject the quantities calculated at~\eqref{Paris} in~\eqref{Vers_ailles}, we obtain the following expression for the gradient of $U^d$ at the ring $\mR_r$:

\begin{equation}\label{Versailles}\begin{split}
\nabla_{x_N}&U^d(\mR_r)=\frac{B+B_0}{2^{\beta+2}\,r^{\beta+1}}\sum_{j=1}^{N-1}\frac{1}{\big|\sin\big(\frac{j\pi}{N}\big)\big|^{\beta+2}}\Bigg[\sin\Big(\frac{2j\pi}{N}\Big)\begin{pmatrix}
    -\sin\left(\frac{2j\pi}{N}\right)\\
    \cos\left(\frac{2j\pi}{N}\right)\\
    0
\end{pmatrix}+\sin\Big(\frac{2j\pi}{N}\Big)\begin{pmatrix}
    0\\
    1\\
    0
\end{pmatrix}\\&-\frac{\beta B_0}{B+B_0}\cos\Big(\frac{2j\pi}{N}\Big)\begin{pmatrix}
    1-\cos\left(\frac{2j\pi}{N}\right)\\
    -\sin\left(\frac{2j\pi}{N}\right)\\
    0
\end{pmatrix}+(\beta+2)\frac{\sin^2\left(\frac{2j\pi}{N}\right)}{2-2\cos\left(\frac{2j\pi}{N}\right)}\begin{pmatrix}
    1-\cos\left(\frac{2j\pi}{N}\right)\\
    -\sin\left(\frac{2j\pi}{N}\right)\\
    0
\end{pmatrix}\Bigg].
    \end{split}
\end{equation}
We now denote by $A_j$ the term in the summation in~\eqref{Versailles} and we gather the terms $A_j$ and $A_{N-j}$ to exploit trigonometric properties. We observe that:
\begin{equation}
\begin{split}
    &\qquad A_j+A_{N-j}\\&=\frac{1}{\big|\sin\big(\frac{j\pi}{N}\big)\big|^{\beta+2}}
    \Bigg[\sin\Big(\frac{2j\pi}{N}\Big)\begin{pmatrix}
    -\sin\left(\frac{2j\pi}{N}\right)\\
    \cos \left(\frac{2j\pi}{N}\right)\\
    0
\end{pmatrix}+\sin\Big(\frac{2(N-j)\pi}{N}\Big)\begin{pmatrix}
    -\sin\left(\frac{2(N-j)\pi}{N}\right)\\
    \cos \left(\frac{2(N-j)\pi}{N}\right)\\
    0
\end{pmatrix}
\\&\qquad+\sin\Big(\frac{2j\pi}{N}\Big)\begin{pmatrix}
    0\\
    1\\
    0
\end{pmatrix}+\sin\Big(\frac{2(N-j)\pi}{N}\Big)\begin{pmatrix}
    0\\
    1\\
    0
\end{pmatrix} -\frac{\beta B_0}{B+B_0}\cos\Big(\frac{2j\pi}{N}\Big)\begin{pmatrix}
    1-\cos\left(\frac{2j\pi}{N}\right)\\
    -\sin\left(\frac{2j\pi}{N}\right)\\
    0
\end{pmatrix} \\
& \qquad -\frac{\beta B_0}{B+B_0}\cos\Big(\frac{2(N-j)\pi}{N}\Big)\begin{pmatrix}
    1-\cos\left(\frac{2(N-j)\pi}{N}\right)\\
    -\sin\left(\frac{2(N-j)\pi}{N}\right)\\
    0
\end{pmatrix}
+(\beta+2)\frac{\sin^2\left(\frac{2j\pi}{N}\right)}{2-2\cos\left(\frac{2j\pi}{N}\right)}\begin{pmatrix}
    1-\cos\left(\frac{2j\pi}{N}\right)\\
    -\sin\left(\frac{2j\pi}{N}\right)\\
    0
\end{pmatrix} \\
& \qquad +(\beta+2)\frac{\sin^2\left(\frac{2(N-j)\pi}{N}\right)}{2-2\cos\left(\frac{2(N-j)\pi}{N}\right)}\begin{pmatrix}
    1-\cos\left(\frac{2(N-j)\pi}{N}\right)\\
    -\sin\left(\frac{2(N-j)\pi}{N}\right)\\
    0
\end{pmatrix}\Bigg]
\end{split}
\end{equation}
One can check in the equation above that we gathered on each line the pairs of term in $A_j$ and $A_{N-j}$ that ``match together'' in the sense that we have cancellations and simplifications using $\cos(2(N-j)\pi/N)=\cos(2j\pi/N)$ and $\sin(2(N-j)\pi/N)=-\sin(2j\pi/N)$. This leads to: 
\begin{equation}
\begin{split}
A_j+A_{N-j}
   & =\frac{2}{\big|\sin\big(\frac{j\pi}{N}\big)\big|^{\beta+2}}
    \Bigg[-\sin^2\Big(\frac{2j\pi}{N}\Big) -\frac{\beta B_0}{B+B_0}\cos\Big(\frac{2j\pi}{N}\Big)\bigg(1-\cos\left(\frac{2j\pi}{N}\right)\bigg) \\
    & \qquad \qquad \qquad \qquad \qquad \qquad \qquad \qquad  \qquad \qquad \qquad \qquad
+\frac{\beta+2}{2}\sin^2\left(\frac{2j\pi}{N}\right)\Bigg]
\begin{pmatrix}
    1\\
    0\\
    0
\end{pmatrix} \\ 
& =\frac{2}{\big|\sin\big(\frac{j\pi}{N}\big)\big|^{\beta+2}}
    \bigg[\frac{\beta}{2}\sin^2\Big(\frac{2j\pi}{N}\Big)+\frac{\beta B_0}{B+B_0}\cos^2\Big(\frac{2j\pi}{N}\Big)-\frac{\beta B_0}{B+B_0}\cos\Big(\frac{2j\pi}{N}\Big)\bigg]\begin{pmatrix}
    1\\
    0\\
    0
\end{pmatrix}
\end{split}
\end{equation}
With $\cos^2(a)=1-\sin^2(a)$ we get:
\begin{equation}\begin{split}
A_j+A_{N-j}&=\frac{2}{\big|\sin\big(\frac{j\pi}{N}\big)\big|^{\beta+2}}
    \bigg[\bigg(\frac{\beta}{2}-\frac{\beta B_0}{B+B_0}\bigg)\sin^2\Big(\frac{2j\pi}{N}\Big)+\frac{\beta B_0}{B+B_0}\bigg(1-\cos\Big(\frac{2j\pi}{N}\Big)\bigg)\bigg]\begin{pmatrix}
    1\\
    0\\
    0
\end{pmatrix}
\end{split}
\end{equation}
We now use the formulas $1-\cos(2a)=2\,\sin^2(a)$ and $\sin(2a)=2\,\sin(a)\,\cos(a)$, and obtain
\begin{equation}\label{Bobigny} 
\begin{split}
&\qquad A_j+A_{N-j} =\frac{2}{\big|\sin\big(\frac{j\pi}{N}\big)\big|^{\beta+2}}
    \bigg[4\bigg(\frac{\beta}{2}-\frac{\beta B_0}{B+B_0}\bigg)\sin^2\left(\frac{j\pi}{N}\right)\cos^2\left(\frac{j\pi}{N}\right) \\
    & \qquad \qquad \qquad +2\frac{\beta B_0}{B+B_0}\sin^2\left(\frac{j\pi}{N}\right)\bigg]\begin{pmatrix}
    1\\
    0\\
    0
\end{pmatrix}\\
&=\frac{4\beta}{\big|\sin\big(\frac{j\pi}{N}\big)\big|^{\beta}}
    \bigg[\bigg(1-\frac{2B_0}{B+B_0}\bigg)\cos^2\left(\frac{j\pi}{N}\right)+\frac{ B_0}{B+B_0}\bigg]\begin{pmatrix}
    1\\
    0\\
    0
\end{pmatrix}
\end{split}
\end{equation} Plugging~\eqref{Bobigny} back into~\eqref{Versailles} eventually gives:
\begin{equation}
    \nabla_{x_N}U^d=\frac{\beta}{2^{\beta+1}\,r^{\beta+1}}\sum_{j=1}^{N-1}\frac{1}{\big|\sin\big(\frac{j\pi}{N}\big)\big|^\beta}
    \bigg[(B-B_0)\cos^2\left(\frac{j\pi}{N}\right)+B_0\bigg]\frac{x_N}{|x_N|}
\end{equation}

\textit{Step 2.} 
One now focuses on the computation of $\nabla_{x_i} U^s$. Similarly, it is enough to compute for $i=N$, due to the $N$-fold symmetry:
\begin{equation}
    \nabla_{x_N} U^s=-\sum_{j=1}^{N-1}\frac{\alpha A}{r^{\alpha+1}\Big|2\sin\left(\frac{j\pi}{N}\right)\Big|^{\alpha+2}}\begin{pmatrix}
    1-\cos\left(\frac{2j\pi}{N}\right)\\
    -\sin\left(\frac{2j\pi}{N}\right)\\
    0
\end{pmatrix}
\end{equation}
Using a symmetrization similar as before and $1-\cos(2a)=2\,\sin^2(a)$, one obtains
\begin{equation}\begin{split}
    \nabla_{x_N} U^s&=-\frac{1}{2}\sum_{j=1}^{N-1}\frac{\alpha A}{r^{\alpha+1}\Big|2\sin\left(\frac{j\pi}{N}\right)\Big|^{\alpha+2}}\Bigg(\begin{pmatrix}
    1-\cos\left(\frac{2j\pi}{N}\right)\\
    -\sin\left(\frac{2j\pi}{N}\right)\\
    0
\end{pmatrix}+\begin{pmatrix}
    1-\cos\left(\frac{2(N-j)\pi}{N}\right)\\
    -\sin\left(\frac{2(N-j)\pi}{N}\right)\\
    0
\end{pmatrix}\Bigg)\\
&=-\sum_{j=1}^{N-1}\frac{\alpha A}{r^{\alpha+1}\Big|2\sin\left(\frac{j\pi}{N}\right)\Big|^{\alpha+2}}\begin{pmatrix}
    1-\cos\left(\frac{2j\pi}{N}\right)\\
    0\\
    0
\end{pmatrix}=\frac{\alpha A}{2^{\alpha + 1} r^{\alpha+1}}\sum_{j=1}^{N-1}\frac{1}{\Big|\sin\left(\frac{j\pi}{N}\right)\Big|^{\alpha}}\begin{pmatrix}
    -1\\
    0\\
    0
\end{pmatrix}
\end{split}
\end{equation}

\textit{Step 3.} We now compute the gradient of $U^d$ with respect to the magnetization $m_i$. We perform the calculations without the constraint $|m_i|=1$.
As before, it is enough to study the gradient with respect to $m_N$ and deduce the other gradients by symmetry.
We write using~\eqref{eq:interaction potential}:
\begin{equation}\begin{split}
    &\qquad\nabla_{m_N}U^d=\sum_{j=1}^{N-1}\bigg(B_0\frac{m_j}{|r_{Nj}|^\beta}-(B+B_0)\frac{(m_j\cdot r_{Nj})r_{Nj}}{|r_{Nj}|^{\beta+2}}\bigg)\\
    &=\sum_{j=1}^{N-1}\Bigg(\frac{B_0}{2^\beta r^\beta\big|\sin\Big(\frac{j\pi}{N}\Big)\big|^\beta}\begin{pmatrix}
    -\sin\left(\frac{2j\pi}{N}\right)\\
    \cos\left(\frac{2j\pi}{N}\right)\\
    0
\end{pmatrix}+\frac{(B+B_0)\sin\Big(\frac{2j\pi}{N}\Big)}{2^{\beta+2}r^\beta\big|\sin\Big(\frac{j\pi}{N}\Big)\big|^{\beta+2}}\begin{pmatrix}
    1-\cos\left(\frac{2j\pi}{N}\right)\\
    -\sin\left(\frac{2j\pi}{N}\right)\\
    0
\end{pmatrix}\Bigg)\end{split}
\end{equation}
Similarly as before, we symmetrize the sum to gather the terms with index $j$ and $N-j$:
\begin{equation}\begin{split}
    \nabla_{m_N}U^d
    &=\frac{1}{2}\sum_{j=1}^{N-1}\Bigg(\frac{B_0}{2^{\beta}r^\beta\big|\sin\Big(\frac{j\pi}{N}\Big)\big|^\beta}\begin{pmatrix}
   0\\
    2\cos\left(\frac{2j\pi}{N}\right)\\
    0
\end{pmatrix}+\frac{(B+B_0)\sin\Big(\frac{2j\pi}{N}\Big)}{2^{\beta+2}r^\beta\big|\sin\Big(\frac{j\pi}{N}\Big)\big|^{\beta+2}}\begin{pmatrix}
    0\\
    -2\sin\left(\frac{2j\pi}{N}\right)\\
    0
\end{pmatrix}\Bigg)\\
&=\sum_{j=1}^{N-1}\Bigg(\frac{B_0\cos\left(\frac{2j\pi}{N}\right)}{2^\beta r^\beta\big|\sin\Big(\frac{j\pi}{N}\Big)\big|^\beta}-\frac{(B+B_0)\sin^2\Big(\frac{2j\pi}{N}\Big)}{2^{\beta+2}r^\beta\big|\sin\Big(\frac{j\pi}{N}\Big)\big|^{\beta+2}}\Bigg)\begin{pmatrix}
    0\\
    1\\
    0
\end{pmatrix}
\end{split}
\end{equation}
Using again the identities $1-\cos(2a)=2\,\sin^2(a)$ and $\sin(2a)=2\,\sin(a)\,\cos(a)$ gives:
\begin{equation}\begin{split}
\nabla_{m_N}U^d
    &=\sum_{j=1}^{N-1}\Bigg(\frac{B_0\Big(1-2\sin^2\Big(\frac{j\pi}{N}\Big)\Big)}{2^\beta r^\beta\big|\sin\Big(\frac{j\pi}{N}\Big)\big|^\beta}-\frac{4(B+B_0)\sin^2\Big(\frac{j\pi}{N}\Big)\cos^2\Big(\frac{j\pi}{N}\Big)}{2^{\beta+2}r^\beta\big|\sin\Big(\frac{j\pi}{N}\Big)\big|^{\beta+2}}\Bigg)\begin{pmatrix}
    0\\
    1\\
    0
\end{pmatrix}\\
&=\sum_{j=1}^{N-1}\Bigg(\frac{B_0-(B+B_0)\cos^2\Big(\frac{j\pi}{N}\Big)}{2^\beta r^\beta\big|\sin\Big(\frac{j\pi}{N}\Big)\big|^\beta}-\frac{2B_0}{2^{\beta} r^\beta\big|\sin\Big(\frac{j\pi}{N}\Big)\big|^{\beta-2}}\Bigg)\begin{pmatrix}
    0\\
    1\\
    0
\end{pmatrix}
\end{split}
\end{equation}
Using now $\cos^2(a)=1-\sin^2(a)$:
\begin{equation}
\nabla_{m_N}U^d
    =\frac{1}{2^\beta r^\beta}\sum_{j=1}^{N-1}\Bigg(\frac{B-B_0}{\big|\sin\Big(\frac{j\pi}{N}\Big)\big|^{\beta-2}}-\frac{B}{\big|\sin\Big(\frac{j\pi}{N}\Big)\big|^\beta}\Bigg)\frac{m_N}{|m_N|}.  
\end{equation}
\end{proof}

\bibliographystyle{plain}
\bibliography{bibliography}

\end{document}

%% file: commands.tex
\usepackage[utf8]{inputenc}
\usepackage{amsmath, amssymb, amsthm}            
\usepackage{amstext, amsfonts, a4}
\usepackage{mathrsfs}
\usepackage{mathtools}
\usepackage[english]{babel}
\usepackage[dvipsnames]{xcolor}
\usepackage{bbm}
\usepackage{stmaryrd}
\usepackage{caption, subcaption}
\usepackage{multicol}

\usepackage{authblk}

\usepackage{setspace}
\makeatletter
\let\oldaffillist\AB@affillist
\renewcommand{\AB@affillist}{\begin{flushleft}\begin{spacing}{1.5}\oldaffillist\end{spacing}\end{flushleft}}
\makeatother

\usepackage{mathtools}
\mathtoolsset{showonlyrefs}

\usepackage[margin=2.2cm]{geometry}
\usepackage{dsfont}

\usepackage[hidelinks]{hyperref}

\allowdisplaybreaks
\numberwithin{equation}{section} 

\newcommand*\diff{\mathop{}\!\mathrm{d}}

\DeclareFontFamily{U}{mathx}{}
\DeclareFontShape{U}{mathx}{m}{n}{<-> mathx10}{}
\DeclareSymbolFont{mathx}{U}{mathx}{m}{n}
\DeclareMathAccent{\widehat}{0}{mathx}{"70}
\DeclareMathAccent{\widecheck}{0}{mathx}{"71}

\theoremstyle{plain}
	\newtheorem{theorem}{Theorem}[section]
	\newtheorem{proposition}[theorem]{Proposition}    
	\newtheorem{lemma}[theorem]{Lemma}
    \newtheorem{remark}[theorem]{Remark}
	\newtheorem{corollary}[theorem]{Corollary}

\theoremstyle{definition}

  \def\mR{{\mathfrak R}} 
\def\mS{{\mathfrak S}}

 \def\NN{{\mathbb N}}  
  \def\RR{{\mathbb R}} 
\def\SS{{\mathbb S}}

\newcommand{\m}[1]{\mathbb{#1}}

\renewcommand{\le}{\leqslant}
\renewcommand{\ge}{\geqslant}